%% file: manuscript.tex
\documentclass[12pt]{article}

\usepackage{natbib}
\usepackage[table]{xcolor}
\input{preamble}

\usepackage{fancyhdr}
\usepackage{abstract}
\usepackage[auth-sc,affil-sl]{authblk}
\usepackage[margin=0.97in]{geometry}
\usepackage{amssymb}
\usepackage{algpseudocode}

\newtheorem{theorem}{Theorem}
\newtheorem{proposition}{Proposition}

\newcommand{\tr}[1]{\text{tr}\bracks{#1}}

\newcommand{\ourtitle}{Optimal Designs with Robust Inference for Binary Treatment Effects}
\title{\ourtitle}

\author[1]{David Azriel\thanks{Electronic address: \texttt{davidazr@technion.ac.il}}}
\author[2]{Abba M. Krieger\thanks{Electronic address: \texttt{krieger@wharton.upenn.edu}}}
\author[3]{Adam Kapelner\thanks{Electronic address: \texttt{kapelner@qc.cuny.edu}; Principal Corresponding author}}

\affil[1]{\small Faculty of Data and Decision Sciences, The Technion, Haifa, Israel}
\affil[2]{\small Department of Statistics, The Wharton School of the University of Pennsylvania, USA}
\affil[3]{\small Department of Mathematics, Queens College, CUNY, USA}

\begin{document}
\maketitle

\begin{abstract}
We study randomized experiments with binary outcomes under Neyman’s nonparametric model, where covariate measurements are fixed but potential outcomes are random. In this setting we derive the exact variance of the difference‑in‑means estimator and characterize designs that minimize it. We show that any balanced design satisfying a covariate‑balance condition is asymptotically optimal, and we prove that a broad class of blocking designs satisfies this condition under mild smoothness assumptions. Because the variance depends on unknown success probabilities, unbiased variance estimation is impossible. We therefore develop two conservative estimators: a generalization of the Cochran–Mantel–Haenszel (CMH) statistic applicable to any balanced design, and an extension of Robins’ variance estimator for blocking designs. We establish conditions under which the CMH‑based estimator is asymptotically tight under local alternatives, thereby yielding asymptotically valid confidence intervals. Our theoretical and simulations results show that blocking and other designs that achieve covariate balance and sufficient degree of randomness, perform well when they are equipped with the CMH-based inference. Thus, we provide an experimental design and inference framework for incidence outcomes that is simultaneously variance‑optimal, conservative in finite samples and asymptotically tight.
\end{abstract}

\vfill
\pagebreak

\section{Introduction}\label{sec:introduction}

We consider an experiment with $n$ subjects and two arms, which we call treatment and control. Subjects are measured on $p$ covariates denoted $\x_i \in \reals^p$ for the $i$th subject where $i=1,\dots,n$. The response of interest is binary and is denoted by $\Y = \bracks{Y_1, \ldots, Y_{n}}^\top \in \braces{0, 1}^n$. The parameter of interest is the difference between the mean success under treatment and under control. The setting we investigate is one in which all $\x_i$'s are seen beforehand and thus called \textit{observed covariates}. Based on these $\x_i$'s, the experimenter chooses the allocation vector $\w = \bracks{w_1, \ldots, w_{n}}^\top \in \braces{-1,+1}^n$, whose entries indicate whether the subject is placed in the treatment arm (coded numerically as +1 and sometimes denoted as T) or  in the control arm (coded numerically as -1 and sometimes denoted as C). Thus, $\w$ is realized from a generalized multivariate Bernoulli random variable $\W$, which is termed a randomized experimental \emph{design}. This design is selected by the experimenter. We consider \qu{balanced} designs where the number of subjects allocated to control and treatment is equal, i.e., $\prob{\sum_{i=1}^n W_i=0}=1$ and thus we assume even $n$. Further, each subject is equally likely to be assigned to treatment and control, i.e., $\expe{\W}=\zerovec_n$. The canonical example of such a design is called the balanced complete randomization design (BCRD), where $\W$ is randomly chosen from the set $\{\w\in\{-1,+1\}^n : \sum_{i=1}^n w_i =0 \}$.

We adopt the potential outcome framework, in which $\Y_T = \bracks{Y_{T,1}, \ldots, Y_{T,n}}^\top$ and $\Y_C = \bracks{Y_{C,1}, \ldots, Y_{C,n}}^\top$ denote the potential  {outcomes} under treatment and control, respectively. For each subject $i$,  {the} observed response $Y_i$ is either $Y_{T,i}$ if $W_i=+1$ or $Y_{C,i}$ if $W_i=-1$.

In order to define the design problem and the resulting statistical inference, one needs to decide which quantities are considered fixed (conditioned upon) and which {are} random. The design $\W$ is always considered random. In the literature, there are typically three scenarios, which are listed below, ordered by degree of randomness.

\begin{table}[h]
\centering
\begin{tabular}{c|ccc}
& \multicolumn{3}{c}{Random?} \\
Setting & Design & Potential Outcomes & Covariates \\
\hline
(a) & $\checkmark$ & &  \\
(b) & $\checkmark$ & $\checkmark$ &  \\
(c) & $\checkmark$ & $\checkmark$ & $\checkmark$
\end{tabular}
\end{table}


\noindent Setting (a) is called the Neyman model (e.g., \citealp{Freedman2008}) or the Neyman-Rubin model (e.g., \citealp{Sekhon2008}) and is widely studied (see \citealp{Imbens2015} and references therein).   
In the context of binary  {responses,} \cite{Robins1988} suggested a conservative confidence interval that is shorter than the standard Wald's interval and is valid asymptotically. More recently, several works  studied confidence intervals that are exact for finite $n$ \citep{Rigdon2015, Li2016, Aronow2023, Li2025exact}.

The Neyman model of Setting (a) is restrictive in that it assumes that the potential outcomes are deterministic functions of the covariates. That is, given the covariates, the potential binary responses are fixed. In the context of clinical trials, it seems unlikely that a positive response is determined only by the treatment assignment and the covariates. A more reasonable model is to assume that the treatment assignment and the covariates determine the probability of success or failure, as in Settings (b) and (c), rather than the actual response as in Setting (a). This is more realistic as there are certainly unobserved covariates that influence the response making the response random even when knowing $\x_i$ and $w_i$.

Consider Setting (c) and assume that $\{(Y_{T,i},Y_{C,i},\x_i)\}_{i=1,\ldots,n}$ are i.i.d. from some distribution. We define the parameter of interest to be the average treatment effect (ATE), $\frac{1}{n}\sum_{i=1}^n \expe{Y_{T,i}-Y_{C,i}}=\expe{Y_{T,1}}-\expe{Y_{C,1}}$ and we estimate it via the difference-in-means estimator expressed a number of ways

\bneqn\label{eq:tau_hat}
	\hat{\tau} := \Ybar_T - \Ybar_C = \frac{1}{n/2}\sum_{\braces{i:W_i=+1}} Y_i ~- \frac{1}{n/2}\sum_{\braces{i:W_i=-1}} Y_i = \frac{2}{n} \W^\top \Y. 
\eneqn

\noindent Under an equally-likely allocation design, $\hat{\tau}$ is unbiased. One can then derive its variance as

\begin{equation}\label{eq:var_c}
\var{\hat{\tau}} = \frac{2}{n}\parens{{p}_T(1-{p}_T)+{p}_C(1-{p}_C)},
\end{equation}

\noindent where ${p}_T:=\expe{Y_{T,i}}$ and ${p}_C:=\expe{Y_{C,i}}$. The derivation can be found in Section~\ref{app:variance_computation} of the Supplementary Material. Equation~\ref{eq:var_c} implies that under Setting (c), the variance does not depend on the design $\W$ (as long {as} it is balanced) and thus there is no difference in designs under Setting (c). This setting is further problematic as it assumes we average over all potential experiments with all $n$ samples of subjects from the population. We should be concerned with designing the best experiment for the $n$ subjects at hand.

Setting (c) is the most random setting, and Equation~\ref{eq:var_c} implies that the standard Wald's confidence interval is asymptotically valid under this setting. As mentioned above, \cite{Robins1988} constructed a valid confidence interval that is shorter than Wald's under Setting (a). This is not surprising because Setting (a) is less random than (c) and hence the variance of the estimator under Setting (a) is smaller. Below, in Section~\ref{sec:Robbins}, we study an extension of Robins' variance estimator to allow it to be valid under Setting (b).

Setting (a) is restrictive and Setting (c) is too random. The focus in this work is on Setting (b), where only the covariates are conditioned upon. A typical approach is to assume a parametric model for $p_{T,i}:=\cexpe{Y_{T,i}}{\x_i}$ and $p_{C,i}:=\cexpe{Y_{C,i}}{\x_i}$, canonically a linear logistic positive outcome model, and estimate its parameters (see \citealp{Imbens2015}, Chapter 8 and references therein). However, as shown by \cite{Freedman2008}, this approach fails if the model is misspecified. Therefore, \cite{Freedman2008} advocates for considering Setting (b) and calls it \qu{Neyman’s nonparametric setup}. This setting is rarely studied in the statistical literature but is more common in econometrics (see \citealp{Bai2024} and references therein). Since we consider Setting (b) from now on, all expectations and variances are conditioned on the $\x_i$'s unless stated otherwise. The ATE can then be expressed as 

\bneqn\label{eq:ate_parameter}
\tau:=\frac{1}{n}\sum_{i=1}^n(p_{T,i}-p_{C,i})
\eneqn

\noindent and a typical estimator is $\hat{\tau}$ of Equation~\ref{eq:tau_hat}. 

Our goal herein is twofold: we wish to minimize $\var{\hat{\tau}}$ and to accurately (or conservatively) estimate the variance. 
We derive a formula for $\var{\hat{\tau}}$ and show that any design that satisfies a balanced-covariate condition is optimal in the sense that it asymptotically minimizes the variance. We further show that a family of blocking designs satisfies this condition.  

The variance of ${\hat{\tau}}$ is generally a complicated function of $p_{T,1},\ldots, p_{T,n},p_{C,1},\ldots,p_{C,n}$ and {an} unbiased estimator of $\var{\hat{\tau}}$ does not exist. Rather, as in \cite{Robins1988}, we construct an estimator $\widehat{\var{\hat{\tau}}}$ that is conservative, i.e., 
$\expenostr{\widehat{\var{\hat{\tau}}}}\ge \var{\hat{\tau}}$.
Our estimator for the variance comes from the Cochran-Mantel-Haenszel test statistic \citep{Cochran1954,Mantel1959}. 
Furthermore, we show that under local alternatives $\widehat{\var{\hat{\tau}}}$ is asymptotically tight under some conditions, which are satisfied by {blocking} designs. The final result is an optimal design and an inference method such that $\var{\hat{\tau}}$ is minimal asymptotically, {is} conservative for finite $n$ and $\widehat{\var{\hat{\tau}}}$ is tight asymptotically.

\section{Theoretical Results} 

\subsection{Setup}\label{sec:setting}

In Setting (b), the potential outcomes and the design are random and the observed covariates are fixed. The observed covariates determine $\p_T:=[p_{T,1},\ldots,p_{T,n}]^\top$ and $\p_C:=[p_{C,1},\ldots,p_{C,n}]^\top$ which are thus fixed and unknown. The potential outcomes are $Y_{T,i} \sim {\rm Bernoulli}(p_{T,i})$ and $Y_{C,i} \sim {\rm Bernoulli}(p_{C,i})$ and all these random variables are assumed independent. The responses $\Y$ can be conveniently expressed as a function of the potential outcomes and the design

\begin{equation}\label{eq:Y_obs}
\Y = \frac{\Y_T+\Y_C}{2}+ \W \odot \frac{\Y_T-\Y_C}{2}
\end{equation}

\noindent where $\odot$ denotes entry-wise product. Our parameter of interest is the ATE (Equation~\ref{eq:ate_parameter}) and the estimator we wish to investigate is $\hat{\tau}$ (Equation~\ref{eq:tau_hat}). In \cite{Kapelner2023} we showed that if the design is balanced{,} then $\hat{\tau}$ is unbiased for $\tau$ and we proved that


\begin{equation}\label{eq:var}
	\var{\hat{\tau}}=\frac{1}{n^2}\parens{(\p_T+\p_C)^\top \bSigmaw (\p_T+\p_C) +2 [\p_T^\top(\onevec -\p_T) + \p_C^\top(\onevec -\p_C)]   },  
\end{equation}

\noindent where $\bSigmaw:=\expe{\W \W^\top}$ as we only examine designs that satisfy $\expe{\W} = \zerovec_n$. It then follows that

\begin{equation}\label{eq:var_LB}
	\var{\hat{\tau}} \ge \frac{2}{n^2}[\p_T^\top(\onevec -\p_T) + \p_C^\top(\onevec -\p_C)].  
\end{equation}

\noindent Further, if we assume that $\p_T$ and $\p_C$ are bounded away from $\onevec_n$ and $\zerovec_n$ respectively, then the lower bound in Equation~\ref{eq:var_LB} is $O(1/n)$ and any design that satisfies

\begin{equation}\label{eq:balance_condition}
	\lim_{n \to \infty} \frac{1}{n}(\p_T+\p_C)^\top \bSigmaw (\p_T+\p_C) =0 
\end{equation}

\noindent is asymptotically optimal. We now show that a family of blocking designs {satisfies} Condition \ref{eq:balance_condition} under mild conditions. 

\subsection{Results for Blocking Designs}


We consider \textit{uniform} block designs with $B$ blocks where all blocks are uniformly sized with $n_B := n/B$ subjects. We assume $n_B$ to be even for convenience. The \textit{optimal} uniform block design partitions $\{1,\ldots,n\}$ into $B$ equally sized disjoint blocks with {indices} $I_1,\ldots,I_B$ such that the sum of the $B$ intrablock \textit{homogeneity metrics} are minimized. Such a homogeneity metric could be the variance, i.e., $\sum_{i \in I_b}\|\x_i - \bar{\x}_b\|^2$ where $\bar{\x}_b:=\frac{1}{n_B}\sum_{i \in I_b} \x_i$; this metric is assumed in Theorem \ref{thm:blocking_balanced} below. This is equivalent to minimization of the squared Euclidean norm between a random pair in a random block. When the $\x$'s are whitened this is equivalent to minimization of the total Mahalanobis distance, which is defined later in Section~\ref{sec:simulations_setup}.

Solving for the optimal uniform blocks is called \textit{balanced clustering}, and there are efficient algorithms to solve it \citep{Malinen2014}. In the special case where the block size is two, the problem can be solved by the nonbipartite matching algorithm  \citep{Edmonds1965}. Then, BCRD is run within each block, i.e., $n_B/2$ subjects are chosen at random and  {allocated} to T and the rest $n_B/2$ subjects are allocated to C. It follows that $\bSigmaw$ is block diagonal consisting of $B$ blocks of size $n_B \times n_B${,} where each block consists of entries with 1 on the diagonal and $-\oneover{n_B - 1}$ on the off-diagonal. 

Returning to Condition \ref{eq:balance_condition}{,} if we denote $v_i:=\frac{p_{T,i}+p_{C,i}}{2}$, then
\[
\frac{1}{n}(\p_T+\p_C)^\top \bSigmaw (\p_T+\p_C)=\frac{4}{B}\sum_{b=1}^B \frac{1}{n_B-1} \sum_{i \in I_b}(v_i-\bar{v}_b)^2,
\]
where $\bar{v}_b$ is the average of the $b$th block. It follows that if the average variance within {a block} converges to zero, Condition~\ref{eq:balance_condition} is satisfied. Theorem \ref{thm:blocking_balanced} below shows that this is satisfied under mild conditions when the number of blocks diverges. The proof is found in Section~\ref{app:proof_thm_1} of the Supplementary Material.

\begin{theorem}\label{thm:blocking_balanced}
	Suppose that $\|\x_i\|$ is bounded for all $i$, and there exist functions $h_T$ and $h_C$ such that $p_{T,i}=h_T(\x_i)$ and $p_{C,i}=h_C(\x_i)$ for all $i${,} and that $h_T$ and $h_C$ are Lipschitz continuous functions. Then, Equation~\ref{eq:balance_condition}  is satisfied for all optimal blocking designs that minimize the block variance with $B \to \infty$.    
\end{theorem}

Assuming that  the Lipschitz constant is 1 and that the support of the $\x_i$'s is contained in the p-dimensional unit cube, the proof of Theorem \ref{thm:blocking_balanced} shows that

\beqn
\frac{1}{n}(\p_T+\p_C)^\top \bSigmaw (\p_T+\p_C) \le 8p\parens{\frac{1}{(B^{1/2p} -1)^2}+\frac{1}{\sqrt{B}}},
\eeqn

\noindent and this bound converges to zero as $B\to \infty$. A similar bound can be obtained if the number of blocks is of different sizes but the ratio of the maximal and minimal block size is bounded. Thus, the result of Theorem \ref{thm:blocking_balanced} can be extended to non-uniform blocking designs provided that the latter condition is satisfied.


\subsection{Estimator Variance Estimation} 

\subsubsection{The CMH Estimator} \label{sec:est_var}

The variance of $\hat{\tau}$ is given in Equation~\ref{eq:var} and is a complicated {function} of the unknown $\p_T$ and $\p_C$, which cannot be estimated. Therefore, an unbiased estimator of $\var{\hat{\tau}}$ does not exist. Following \citet{Neyman1990,Robins1988} and others, we aim to construct a conservative estimator $\widehat{\var{\hat{\tau}}}$ in the sense that $\mathbb{E}[\widehat{\var{\hat{\tau}}}] \ge \var{\hat{\tau}}$. In Section \ref{sec:local}, we show that our estimator is asymptotically tight under  {a} local alternative {asymptotic regime}. 

Specifically, for a general design $\W$ we define

\bneqn\label{eq:cmh_variance_general}
V_{CMH}:=\frac{4}{n^2} \Y^\top \bSigmaw \Y. 
\eneqn

\noindent We showed in a previous paper that the above expression is the variance estimator of the Cochran-Mantel-Haenszel test statistic \citep[Equation 3]{azriel2026block}. We also showed therein that for blocking designs we can express the estimator as

\bneqn\label{eq:cmh_variance_blocking}
V_{CMH} = \frac{1}{n^2(n_B - 1)} \sum_{b=1}^B n_{1_b} n_{0_b}
\eneqn

\noindent where $n_{1_b}, n_{0_b}$ are the number of positive and negative responses in block $b$ respectively. This previous paper considers only blocking designs, whereas here we extend our analysis to any balanced design. 

An interesting feature of $V_{CMH}$ is that it does not depend explicitly on $\W$ (only through $\Y$),  {but rather} it depends on {the} covariance matrix of $\W$, i.e., on $\bSigmaw$. This matrix can in principle be estimated accurately by computing the sample variance-covariance matrix over many draws of $\w$ since $\W$ is determined by the experimenter, i.e., $\bSigmaw \approx \frac{1}{R} \sum_{r=1}^R \w_r \w_r ^\top$ where the set $\{\w_1,\ldots,\w_R\}$ is drawn from $\W$. 
In some cases the design is to select uniformly from a set $\{\w_1,\ldots,\w_R\}$. In that case the expression applies and is no longer approximate. In the current work, an important example is blocking designs, in which $\bSigmaw$ is explicitly known (the block diagonal form described above) and thus is more amenable to theoretical exploration.


The next theorem shows that $V_{CMH}$ is conservative. To state the theorem{,} we define $\bSigmaw^{\odot 2}:=\bSigmaw \odot \bSigmaw$ and ${\bseta}:=\frac{1}{2}\parens{\p_T-\p_C}$. The proof is found in Section~\ref{app:thm_2_proof} of the Supplementary Material.

\begin{theorem}\label{thm:exp_var_hat}
Consider Neyman's nonparametric model and suppose that the design is balanced. Then, 
\[
\expe{V_{CMH}}- \var{\hat{\tau}}=\frac{4}{n^2}{\boldsymbol \eta}^\top \bSigmaw^{\odot 2} {\boldsymbol \eta}.
\]
By {the} Schur product theorem  \citep[Theorem 7.5.3]{horn2013matrix}, $\bSigmaw^{\odot 2}$ is positive semi-definite, and therefore $\expe{V_{CMH}} \ge \var{\hat{\tau}}$.
\end{theorem}

Theorem \ref{thm:exp_var_hat} shows that the difference between $\expe{V_{CMH}}$ and $\var{\hat{\tau}}$ is $\frac{4}{n^2}{\boldsymbol \eta}^\top \bSigmaw^{\odot 2} {\boldsymbol \eta}$. This difference is always positive{,} making the variance estimator conservative, but it is of interest to understand when the difference is small. 
The difference is zero if $\bseta={\bf 0}$, i.e., if $p_{T,i}=p_{C,i}$ for all $i\in\{1,\ldots,n\}$. This condition is a slight modification of Fisher's sharp null, where the probabilities of success are equal between treatment and control rather than the actual responses. Thus, $V_{CMH}$ is tight when the difference between $\p_T$ and $\p_C$ is small. In Section \ref{sec:local}{,} we consider {a} local alternatives {asymptotic regime} where $\tau=O(1/\sqrt{n})$ and show that in this regime, $V_{CMH}$ is asymptotically tight under additional conditions on the design.

We next argue informally that the difference term, ${\boldsymbol \eta}^\top \bSigmaw^{\odot 2} {\boldsymbol \eta}$, for blocking designs is decreasing with $n_B$ (the block size). Recall that in blocking designs 
 $\bSigmaw$ is block diagonal{,} consisting of $B$ blocks of size $n_B \times n_B${,} where each block consists of entries with 1 on the diagonal and $-1/(n_B - 1)$ on the off-diagonal. Therefore,
\begin{equation}\label{eq:diff.term.blocking}
{\boldsymbol \eta}^\top \bSigmaw^{\odot 2} {\boldsymbol \eta}=\frac{n_B}{n_B-1} \| {\bseta}\|^2 -\frac{1}{n_B-1}{\boldsymbol \eta}^\top \bSigmaw {\boldsymbol \eta}.    
\end{equation}

The second term in Equation~\ref{eq:diff.term.blocking} is of smaller order than the first because when $n_B$ is small{,} the blocks are homogeneous in $\p_T$ and $\p_C$, making ${\boldsymbol \eta}^\top \bSigmaw {\boldsymbol \eta}$ small{; on the other hand} for large $n_B$, the term $\frac{1}{n_B-1}{\boldsymbol \eta}^\top \bSigmaw {\boldsymbol \eta}$ is small because of $n_B$ in the denominator. It follows that the dominant term in Equation~\ref{eq:diff.term.blocking} is $\frac{n_B}{n_B-1} \| {\bseta}\|^2$, which is decreasing in $n_B$ and therefore so is (approximately) ${\boldsymbol \eta}^\top \bSigmaw^{\odot 2} {\boldsymbol \eta}$. Thus, there is a trade-off between making $\var{\hat{\tau}}$ small and making the difference term small. The former requires small $n_B$ (as in Theorem \ref{thm:blocking_balanced}) while for the latter,  {larger} $n_B$ is preferred. 

The maximal eigenvalue $\lambda_{max}$ of $\bSigmaw$ may represent the randomness or robustness of the design \citep{Harshaw2024}. For blocking designs{,} it is equal to $\frac{n_B}{n_B-1}$, which appears in Equation~\ref{eq:diff.term.blocking}, implying that the design is more random or robust as $n_B$ increases. Thus, we expect for a hormetic U-shape in performance over $n_B$ with an optimal block size between the unrestricted BCRD ($n_B = n$ and $B=1$) and the relatively more restricted pairwise matching design ($n_B = 2$ and $B = n/2$). This U-shape is due a trade-off between covariate imbalance and randomness which seems to be fundamental in experimental design \citep{Kapelner2021,Harshaw2024}.  

A similar trade-off holds for general designs; by Equation~\ref{eq:var}, $\var{\hat{\tau}}$ is small when $(\p_T+\p_C)^\top \bSigmaw (\p_T+\p_C)$ is minimal, which can be achieved when the design balances the covariates between treatment and control. Generally, this is achieved for $\bSigmaw$ when $\lambda_{max}$ is large \citep{Kapelner2021,Harshaw2024}. On the other hand, the difference term can be bounded by the Fan-Horn inequality \citep[][page 312]{Horn_Johnson_1991},

\beqn
{\boldsymbol \eta}^\top \bSigmaw^{\odot 2} {\boldsymbol \eta} \le \lambda_{max}^2 \| {\boldsymbol \eta}\|^2, 
\eeqn

\noindent It follows that the difference can be bounded when $\lambda_{max}$ is small. Thus, again, minimizing $\var{\hat{\tau}}$ requires large $\lambda_{max}$, while the difference is small (or at least upper bounded) when $\lambda_{max}$ is small.

\subsubsection{The Robins (1988) Estimator}\label{sec:Robbins}

As mentioned in the introduction, \citet{Robins1988} suggested a confidence interval {for} a similar problem in Setting (a), when the potential outcomes are fixed under BCRD (he {also considered} non-balanced designs, but these are out of the scope of our work herein). 

We first extend his variance estimator to general blocking designs and then to  Neyman's nonparametric model. In the fixed potential outcomes setting{,}

\begin{equation}\label{eq:var_cond_Y}
\cvar{\hat{\tau}}{\Y_T,\Y_C}=\frac{1}{n^2} (\Y_T+\Y_C)^\top \bSigmaw (\Y_T+\Y_C).    
\end{equation}

\noindent Robins constructs a conservative variance estimator for $\cvar{\hat{\tau}}{\Y_T,\Y_C}$ for BCRD as follows:
\[
V_{Robbins}:=\frac{m_1(1-m_1)}{n/2}+\frac{m_0(1-m_0)}{n/2}+\frac{(2m_0-m_1)(1-m_1)-m_0(1-m_0)}{n},
\]
where $m_1=\max\{\hat{p}_T,\hat{p}_C \}$, $m_0=\min\{\hat{p}_T,\hat{p}_C \}$ and $\hat{p}_T=\frac{\sum_{i=1}^n I(W_i=1) Y_i}{n/2}$, $\hat{p}_C=\frac{\sum_{i=1}^n I(W_i=-1) Y_i}{n/2}$. It is easy to extend his variance estimator to an arbitrary  {blocking} design with $B$ blocks and block size $n_B=n/B$. Due to independence of the blocks, the conservative variance estimator becomes

\footnotesize
\begin{equation} \label{eq:V_Robbins}
V_{Robbins}= \frac{1}{B^2} \sum_{b=1}^B\left[ 
\frac{m_{1,b}(1-m_{1,b})}{n_B/2}
+\frac{m_{0,b}(1-m_{0,b})}{n_B/2}
+\frac{(2m_{0,b}-m_{1,b})(1-m_{1,b})-m_{0,b}(1-m_{0,b})}{n_B} 
\right] 
\end{equation} 
\normalsize

\noindent where $m_{1,b}=\max\{\hat{p}_{T,b},\hat{p}_{C,b} \}$, $m_{0,b}=\min\{\hat{p}_{T,b},\hat{p}_{C,b} \}$ and $\hat{p}_{T,b}:=\oneover{n_B/2}{\sum_{i \in I_b} I(W_i=1) Y_i}$, $\hat{p}_{C,b}:=\oneover{n_B/2}{\sum_{i \in I_b} I(W_i=-1) Y_i}$. 

We now extend this estimator to Neyman's nonparametric model. The variance in the fixed outcome model is given in Equation~\ref{eq:var_cond_Y}. Taking its expectation yields

\beqn
\frac{1}{n^2}\expe{(\Y_T+\Y_C)^\top \bSigmaw (\Y_T+\Y_C)}
=\frac{1}{n^2}\parens{
(\p_T+\p_C)^\top \bSigmaw (\p_T+\p_C)
+\p_T^\top(\onevec -\p_T) 
+ \p_C^\top(\onevec -\p_C)
}.
\eeqn

\noindent Comparing this term to $\var{\hat{\tau}}$ as given in Equation~\ref{eq:var}, we need to add a conservative estimate to the component
\[
\frac{1}{n^2}\parens{\p_T^\top(\onevec -\p_T) + \p_C^\top(\onevec -\p_C)}. 
\]
The function $g(t)=t(1-t)$ is concave and therefore, by Jensen's inequality,
\[
\frac{1}{n^2}\parens{\p_T^\top(\onevec -\p_T) + \p_C^\top(\onevec -\p_C)}
\le \frac{1}{n}\parens{\bar{p}_T(1 -\bar{p}_T) + \bar{p}_C(1 -\bar{p}_C)}, 
\]
where $\bar{p}_T:=\frac{1}{n}\sum_{i=1}^n p_{T,i}$ and $\bar{p}_C:=\frac{1}{n}\sum_{i=1}^n p_{C,i}$, which can be consistently estimated by replacing $\bar{p}_T$ and $\bar{p}_C$ with $\hat{p}_T$ and $\hat{p}_C$. Thus, our final variance estimator is

\bneqn\label{eq:robins_var}
V_{Robbins-ext}:=
V_{Robbins} 
+ \frac{1}{n} \parens{\hat{p}_T(1-\hat{p}_T)+ \hat{p}_C(1-\hat{p}_C)}.
\eneqn

For $V_{CMH}${,} we showed that for each $n$ {it} {is} conservative in the sense that $\expe{V_{CMH}}\ge \var{\tau}$. However, $V_{Robbins-ext}$ is conservative only for large samples, i.e., it is required that $n_B$ {be sufficiently} large to allow consistent estimation of the average $p_{T,i}$'s and $p_{C,i}$'s of each block. Another important difference is that $V_{CMH}$ holds for general balanced designs{,} whereas $V_{Robbins-ext}$ is only for blocking designs.

To illustrate the differences between $V_{CMH}$ and $V_{Robbins-ext}${,} we consider the BCRD design. The following proposition computes the {variance} estimators. Its proof is found in Section~\ref{app:proof_prop_1} of the Supplementary Material.

\begin{proposition}\label{prop:BCRD}
Under the BCRD design,
\[
V_{Robbins-ext}
=\frac{2}{n}\parens{  
\hat{p}_T(1-\hat{p}_T)+ \hat{p}_C(1-\hat{p}_C) + m_0(1-m_1)
}
\]
and
\[
V_{CMH}
=\frac{2}{n-1}\parens{ 
\hat{p}_T(1-\hat{p}_T)+ \hat{p}_C(1-\hat{p}_C) 
+ \frac{ (\hat{p}_T - \hat{p}_C)^2}{2}
}.
\]
\end{proposition}

\noindent If we make the approximation $n\approx n-1$, then  $V_{CMH} \le V_{Robbins-ext}$ iff $(\hat{p}_T - \hat{p}_C)^2 \le 2m_0(1-m_1)$. Recall that $m_1=\max\{\hat{p}_T,\hat{p}_C \}$ and $m_0=\min\{\hat{p}_T,\hat{p}_C \}$. Thus, if $(\hat{p}_T - \hat{p}_C)^2$ is sufficiently small, $V_{CMH} \le V_{Robbins-ext}$. For illustration, suppose that $\hat{p}_T=0.5+a$ and $\hat{p}_C=0.5-a$ where $a\in (0,1/2)$. Then, the inequality reads
\[
4 a^2 \le 2 (0.5-a)^2  
\Longleftrightarrow 
a \le  (\sqrt{2}-1)/2 \approx 0.21.
\]
Thus, $V_{CMH}$ is smaller than $V_{Robbins-ext}$ when $\hat{p}_T$ and $\hat{p}_C$ are close to each other{,} and otherwise $V_{Robbins-ext}$ is smaller. This finding is {consistent with} those of Section \ref{sec:est_var}, namely that $V_{CMH}$ works well when $\p_T \approx \p_C$.

\subsection{Local Asymptotic Analysis for the CMH Estimator}\label{sec:local}

Local alternative asymptotics (often referred to as Pitman alternatives) {provide} a framework for evaluating the power of a test by treating the truth as a ``moving target'' \citep[][Chapter 14]{Van2000}. When the alternative hypothesis remains fixed{,} the power converges to 1, whereas the local alternative regime defines a sequence of alternatives that drift closer to the null hypothesis at rate $1/\sqrt{n}$, and the power converges to a constant. The motivation is that if the alternatives converge {to} the null at a rate faster than $1/\sqrt{n}$, then the alternative is statistically indistinguishable from the null, and if the rate is slower (for example{,} fixed alternatives), then the problem is ``easy''. In this sense, the local alternative asymptotics approximate a situation where the statistical problem is neither too easy nor too difficult.  
In \cite{azriel2026block}{,} we studied the power of the Cochran-Mantel-Haenszel {test under} local alternatives for blocking designs. Here{,} we extend the results to general balanced designs and study variance estimators.

Under Neyman's nonparametric model{,} the parameters are $\p_T$ and $\p_C$. In the local alternative asymptotics{,} we assume that $\p_T$ and $\p_C$ depend on $n$, which is suppressed in the notation, and
\begin{equation}\label{eq:local_alternative}
\sqrt{n} \tau =\frac{1}{\sqrt{n}}\sum_{i=1}^n(p_{T,i}-p_{C,i}) \longrightarrow \tau_\infty,    
\end{equation}
for a constant $\tau_\infty$ as $n \to \infty$. We also assume that 
\begin{equation}\label{eq:moment_condition}
\frac{1}{n}\| \bseta\|^2 =\frac{1}{4n}\sum_{i=1}^n(p_{T,i}-p_{C,i})^2 \longrightarrow 0    
\end{equation}
as $n \to \infty$. Condition~\ref{eq:moment_condition} implies that not only the average of $\p_T-\p_C$ {converges} to zero as in Condition~\ref{eq:local_alternative}, but also the second moment. It can be violated if{,} for example{,} $p_{T,i}-p_{C,i}=C (-1)^i$ for a constant $C$, {for} $i=1,\ldots,n$. When this condition is violated, the consistency result below does not hold. Notice that $\var{\hat{\tau}}$ is of order $1/n$ and therefore, $V_{CMH}$ is tight if $ n \parens{\expe{V_{CMH}}- \var{\hat{\tau}}}$ converges to zero. 
The following theorem shows that $V_{CMH}$ is tight and consistent under Condition~\ref{eq:moment_condition} and additional conditions. Its proof is found in Section~\ref{app:proof_thm3} of the Supplementary Material.

\begin{theorem}\label{thm:local}
Consider Neyman's nonparametric model and suppose that the design is balanced, that Condition~\ref{eq:moment_condition} holds{,} and that $\lambda_{max}(\bSigmaw)$ is bounded. Then
\[
n \parens{\expe{V_{CMH}}- \var{\hat{\tau}}} \longrightarrow 0 , 
\]
as $n \to \infty$. Moreover, if $\max_{i} \braces{\sum_{j} \left| (\bSigmaw)_{i,j}\right|}$ is bounded, then
\[
\var{n V_{CMH}} \longrightarrow 0 , 
\]
and hence $n V_{CMH}$ is a consistent estimator of $n \var{\hat{\tau}}$.
\end{theorem}

Thus, Theorem \ref{thm:local} shows that $V_{CMH}$ is tight if the matrix $\bSigmaw$ is ``small'' in the sense that the eigenvalues are bounded. For blocking designs{,} the maximal eigenvalue is $\frac{n_B}{n_B-1}$, which is maximized when $n_B=2$ (pairwise matching), and is bounded. On the other hand, if we consider the nearly deterministic design {that chooses} at random between $\{+\w_0,-\w_0\}$, where $\w_0$ is an assignment that minimizes a certain {balancing} criterion \citep{Kallus2018}, then $\bSigmaw$ is of rank 1, and $\lambda_{max} = n$. Thus, for these designs{,}  $\lambda_{max}$ is unbounded and $V_{CMH}$ is not tight.

To prove that $n V_{CMH}$ is a consistent estimator{,} we  {need} a stronger condition; namely, that  $\max_{i} \sum_{j} \left| (\bSigmaw)_{i,j}\right|$ is bounded. This is a stronger condition because

\beqn
\lambda_{max} \le \max_{i} \braces{\sum_{j} \left| (\bSigmaw)_{i,j}\right|}
\eeqn

\noindent by a consequence of the Gershgorin circle theorem \citep[Theorem 1.1]{varga2004gershgorin}. For uniform blocking designs, $\max_{i} \sum_{j} \left| (\bSigmaw)_{i,j}\right|=2$, i.e., the {boundedness} condition is satisfied. Notice, that this bound also applies to non-uniform blocking designs.  

Thus, the conditions of Theorem \ref{thm:local} hold for blocking designs and hence they are asymptotically tight under Condition~\ref{eq:moment_condition}. By Theorem \ref{thm:blocking_balanced}, blocking designs with $B \to \infty$ achieve the minimal variance of $\hat{\tau}$. It follows that these designs are both asymptotically optimal and tight. While Theorem \ref{thm:local} shows that the difference term  
$\frac{1}{n^2}{\boldsymbol \eta}^\top \bSigmaw^{\odot 2} {\boldsymbol \eta}$ is asymptotically negligible, for small or moderate sample size it might be noticeable, as shown in the simulation section. 
As we demonstrated in Section \ref{sec:est_var}, the difference term is large for small $n_B$ (see Equation~\ref{eq:diff.term.blocking}). Therefore, practically, one should aim for small $n_B$ but not too small; see Figure \ref{fig:hormesis}.

If {$\hat{\tau}$ satisfies that} it is asymptotically normal under the asymptotic regime of Conditions~\ref{eq:local_alternative} and \ref{eq:moment_condition}, i.e.,
\begin{equation}\label{eq:asymp_normal}
\frac{\sqrt{n}(\hat{\tau} - \tau)}{\sqrt{n\var{\hat{\tau}}}} \convd N(0,1),     
\end{equation}
then for a design that satisfies the conditions of Theorem \ref{thm:local}, we have{,} by Slutsky's theorem{,} that
\[
\frac{\sqrt{n}(\hat{\tau} - \tau)}{\sqrt{n V_{CMH}}} \convd N(0,1), 
\]
and therefore{,} $\hat{\tau}\pm z_{1-\alpha/2} \sqrt{nV_{CMH}}$ is a valid confidence interval for $\tau$ at level $1-\alpha$. In \cite{azriel2026block}{,} we showed that blocking designs satisfy Equation~\ref{eq:asymp_normal} under Conditions~\ref{eq:local_alternative} and \ref{eq:moment_condition} and additional conditions, and therefore one can construct valid confidence intervals using blocking designs in this asymptotic regime.

\section{Simulations}\label{sec:simulations}

\subsection{Setup}\label{sec:simulations_setup}
 
In our discussed inference procedures, we simulate under a wide range of experimental designs, sample sizes and covariate dimensions and compare four metrics: power, coverage, size and confidence interval length. For each combination of sample size $n \in \{64, 128, 256\}$ and covariate dimension $p \in \{1, 2, 5, 10\}$, an $n \times p$ covariate matrix $\mathbf{X}$ is drawn once from the standard multivariate normal distribution and held \emph{fixed} across all replications in the $n \times p$ cell as we are considering setting (b). Within each cell, we simulate under the following $\W$ designs:

\begin{description}

\item[{BCRD}] This design is described in Section~\ref{sec:introduction}. It is the only design considered in this simulation which is blind to the observed covariates.

\item[{Rerandomization}]
Here, allocations are first drawn from a large pool of $100N_{sim}$ \textbf{BCRD} allocations and then the pool is winnowed to those whose  total Mahalanobis Distance (MD) between treated and control covariate vector averages lies within the lowest $1\%$ of Mahalanobis distances. {Total MD is defined as $(\bar{\x}_{T} - \bar{\x}_{C})^\top \hat{\Sigma}_X^{-1} (\bar{\x}_{T} - \bar{\x}_{C})$ where $\bar{\x}_{T}$ and $\bar{\x}_{C}$ denote the average covariate vector over all $n/2$ treatment subjects and $n/2$ control subjects respectively and $\hat{\bSigma}_X$ is sample variance-covariance matrix over all rows in $\X$.}  The $\w$'s in this design are then drawn randomly from this winnowed pool. Designs like these date back to \citet{Student1938}. See \citet{Morgan2012} for modern theoretical analysis. 

\item[{BinaryMatch}]
Subjects are matched into pairs by first generating a $n \times n$ symmetric distance matrix whose entries are the intra-pair MDs \citep[Section 2.2]{Stuart2010} defined as $d(\x_a,\x_b):= (\x_a-\x_b)^\top \hat{\bSigma}_X^{-1} (\x_a-\x_b)$, where $\x_a,\x_b$ are a pair of subjects and $\hat{\bSigma}_X$ is defined in the previous paragraph. We then use the optimal nonbipartite matching algorithm which finds the globally minimum sum of intra-pair MDs computed via the \texttt{nbpMatching} \textsc{R} package \citep{Beck2024nbpMatching}, which implements the shortest-path algorithm of \citet{Derigs1988}. (For $p=1$, the pairs are generated trivially by sorting the one covariate vector and splicing into groups of two). Within pairs, we randomize T/C or C/T with a fair coin flip to produce draws of $\w$. This design is a special case of blocking with $B = n/2$ and $n_B = 2$. 

\item[{GreedyMD}]
Subjects are matched into pairs using an iterative process beginning with a draw of BCRD. All potential pairs of T,C are switched and the total MD is calculated. The switch that produces the lowest total MD is retained. The iterations continue until total MD is no longer reduced. The final $\w$ is the allocation returned. This greedy method produces very low observed balances and are nearly as random as BCRD \citep{Krieger2019}.

\item[{BinaryMatchThenGreedyMD}]
This hybrid approach first finds the binary match structure via the nonbipartite matching algorithm (see description above about the \textit{BinaryMatch} design). We then draw one $\w$ from the \textit{BinaryMatch} design. We then do the greedy search described in the \textit{GreedyMD} section but only switches within the binary match structure are considered. This greedy method produces allocations with even lower observed balances than \textit{GreedyMD} and are nearly as random as \textit{BinaryMatch} \citep{Krieger2022}.

\item[{Naive $B$}]
This is a naive uniform blocking design where subjects are partitioned into exactly $B$ equally-sized blocks by a greedy covariate-stratification algorithm.  Each subject is assigned a composite stratum key formed by concatenating quantile-bin labels across covariates one column at a time: for each covariate in turn, the algorithm selects the largest number of quantile bins that, when crossed with the bins already determined by prior covariates, keeps the total number of distinct keys at or below $B$ while maintaining equal block sizes.  The process stops as soon as exactly $B$ blocks are formed.  Because covariates are processed sequentially and stratum commitments are irrevocable, earlier covariates dominate the block structure and later ones are used only to subdivide remaining blocks. At high $p$, the algorithm typically exhausts the block budget after one or two covariates and ignores the rest. Treatment is then balanced within each block by using \textit{BCRD} within block to draw $\w$.  We consider $B \in \{2,4,8,16,32,64,128\}$, excluding values of $B$ where $n/B < 2$.

\item[{Optimal $B$}]
This is the optimal uniform blocking design discussed in Section~\ref{sec:setting}. Subjets partitioned into exactly $B$ covariate-homogeneous, equally-sized blocks by minimizing the total MD. This minimization is then subject to the constraint that each block contains exactly $n/B$ subjects.  This is a balanced clustering problem which differs from the standard $K$-means problem, which ignores group-size balance. Here, we must consider only uniform block sizes as the space of possible solutions. We employ the \texttt{anticlust} \textsc{R} package \citep{Papenberg2026}, which uses an iterative $K$-means-style exchange algorithm to converge to a locally optimal partition.  Treatment is then balanced within each block by using \textit{BCRD} within block to draw $\w$. We consider the same values of $B$ as described in the \textit{Naive Blocking} description above. 

\end{description}
  
For each of the above designs, we draw $N_{sim} = 10,000$ iid replicates of $\w$ within each $n \times p$ cell. For each $\w$ replicate, we draw $\y$ (which is deterministic per replicate and hence denoted by lower case $\y$)  using independent Bernoulli draws from the probabilities of positive outcome computed from the logistic-linear model, 

\bneqn\label{eq:dgp}
    Y_{i} \mid \mathbf{x}_{i}, w_{i}
      \;\sim\;
    \mathrm{Bernoulli}\!\left(
      \mathrm{expit}\!\left(\beta_0 + \mathbf{x}_{i}^{\top} \bbeta
        + \beta_{T}\,w_{i}\right)
    \right),
\eneqn
  
\noindent where the covariate weights $\bbeta$ is $p$-dimensional, fixed to be even spaced between -1, +1 and then scaled to satisfy $\|\bbeta\| = 3$. This norm was picked to bump the importance of the covariates in the response, thereby magnifying the performance of designs that balance observed covariates well, which ultimately results in better separation in our results between designs. If $\|\bbeta\|$ were to be much larger, $\p_T$ and $\p_C$ would be pushed to $\onevec_n$ and $\zerovec_n$ respectively. 
This situation would correspond with Setting (a) and our theoretical results would not apply. 
For simplicity, we set $\beta_0 = 0$ but there is still a small intercept term since $\X$ is drawn from a standard normal distribution (however this term is $O(n^{-1/2})$ and negligible for the sample sizes considered). We set $\beta_{T} = 0.5$ when measuring power, coverage and CI length; and we set $\beta_{T} = 0$ when measuring size. Because $\X$ and $\bbeta$ are fixed and thus only vary within each $n \times p$ cell, the ATE (Equation~\ref{eq:ate_parameter}) within each cell is calculated as

\beqn
    \tau = \oneover{n}\sum_{i=1}^{n} 
        \parens{\text{expit}\parens{
          \mathbf{x}_{i}\bbeta+\beta_{T})
          -
          \text{expit}(\mathbf{x}_{i} \bbeta-\beta_T
        }}.
\eeqn
 
For each replication pair, we estimate the ATE via the observed $\ybar_T - \ybar_C$ (Equation~\ref{eq:tau_hat}). Then we calculate three different standard errors estimates ($se$) which will allow us to compare the performance of our three procedures: the CMH (the square root of Equation~\ref{eq:cmh_variance_general} in general but employing Equation~\ref{eq:cmh_variance_blocking} instead when $\W$ is blocking), Robins (the square root of Equation~\ref{eq:robins_var}) and $\sqrt{\ybar_T(1-\ybar_T) + \ybar_C(1-\ybar_C)}(n/2)^{-1/2}$, the standard Wald expression. Notice that Proposition \ref{prop:c} implies that the latter $se$ estimate is realized from a conservative variance estimator, as it is a consistent estimator for the variance of $\tauhat$ when the variance is with respect to both the covariates and potential outcomes. This variance is larger than the expected variance given the covariates. This point is one of the main results we wish to convey.

Equation~\ref{eq:cmh_variance_general} requires a closed form expression for $\bSigmaw$ which we believe is unknown for the 
\textit{GreedyMD} and the \textit{BinaryMatchThenGreedyMD} designs. We can estimate $\bSigmaw$ via the Monte Carlo procedure described in Section~\ref{sec:est_var}. Then, using Equation~\ref{eq:cmh_variance_general}, the per-replicate CMH variance estimate can be computed via

\bneqn\label{eq:monte_carlo_se}
\widehat{\var{\tauhat}} \approx \frac{4}{n^2} \y^\top\parens{\frac{1}{N_{sim}} \sum_{r=1}^{N_{sim}} \w_r \w_r^\top} \y =  \frac{4}{n^2}\frac{1}{N_{sim}}\sum_{r=1}^{N_{sim}} (\y^\top\w_r)^2.
\eneqn

\noindent As $\y$ is drawn anew in each replicate, $\widehat{\var{\tauhat}}$ must be calculated anew per replicate which is computationally expensive. This cost can be mitigated by caching the $N_{sim}$ $\w_r$'s in each $n \times p$ cell as these $\w_r$'s only depend only on the fixed $\X$. The closed form expression for $\bSigmaw$ in the \textit{rerandomization} design is also unknown. In our implementation of the \textit{rerandomization} design, it can be considered a fixed set of $N_{sim}$ vectors. Here we can still use Equation~\ref{eq:monte_carlo_se} and it is no longer approximate.


Ultimately, each $n \times p \times~$inferential method results in one $se$ estimate. Using these $se$ values, we then compute a two-sided
$p$-value for $H_{0}: \tau = 0$ at nominal level $\alpha = 5\%$. We then define our four comparison performance metrics as follows: power is the proportion of times that $p\text{-value} < \alpha$ over all $N_{sim}$ replicates when $\beta_T > 0$, size is the proportion of times that $p\text{-value} < \alpha$ over all $N_{sim}$ replicates when $\beta_T = 0$, coverage is the proportion of times $\tau \in \bracks{\ybar_T - \ybar_C \pm 1.96 \times se}$ over all $N_{sim}$ replicates when $\beta_T > 0$ and CI length is the average value of $2 (1.96 \times se)$ over all $N_{sim}$ replicates when $\beta_T > 0$.

\subsection{Results}\label{sec:simulations_results}

Figure~\ref{fig:power_and_size_p_1} shows the power and size results and Figure~\ref{fig:coverage_and_length_p_1} shows the coverage and confidence interval length results for legal inference when $p=1$, all sample sizes, all designs and all inference methods. We include the analogous figures for $p=5$ and $p=10$ in Section~\ref{app:additional_results} of the Supplementary Information. Figure~\ref{fig:power_and_size_p_1} demonstrates that our CMH procedure overall has both higher power than the Wald and Robins procedures and are more appropriately sized. The Wald procedure is conservative which does not seem to improve with sample size. The Robins procedure is wildly missized --- for small $B$ it is too conservative and for large $B$ it is anti-conservative. The latter is a result of inconsistent estimation of $p_{T,b}$ and $p_{C,b}$ for small blocks, as the estimators are based on a small number of observations (see Equation~\ref{eq:V_Robbins}). Figure~\ref{fig:coverage_and_length_p_1} demonstrates tighter confidence bands for CMH in comparison to Wald and Robins which is expected given the higher power. As for coverage, all inference methods are conservative (except the Robins procedure when $n_B = 2$), but the CMH procedure is in general less conservative with confidence bands most often including the desired $1-\alpha$ target. The additional results for $p > 1$ are similar, although the comparisons between the methods are more compressed.

We provide a clearer illustration of CMH's power and confidence interval length performance gains versus the Wald procedure for $n=64$ and $p=1$ in 
Table~\ref{tab:cmh_gains}. Over the different settings, there is a 54-85\% power gain and 16-21\% reduction in the confidence interval length (save the BCRD setting which has virtually equivalent performance).

We also provide a clearer illustration of the CMH estimator's asymptotic hormetic U-shape performance over $n_B$ for the \textit{Optimal Blocking} design in Figure~\ref{fig:hormesis}. The results are most clearly seen in the $p=1$ case where the blocks are optimal. Within simulation error, the U-shape is also plausible for $p>1$.

\begin{figure}[htp]
    \centering
    \includegraphics[width=1\linewidth]{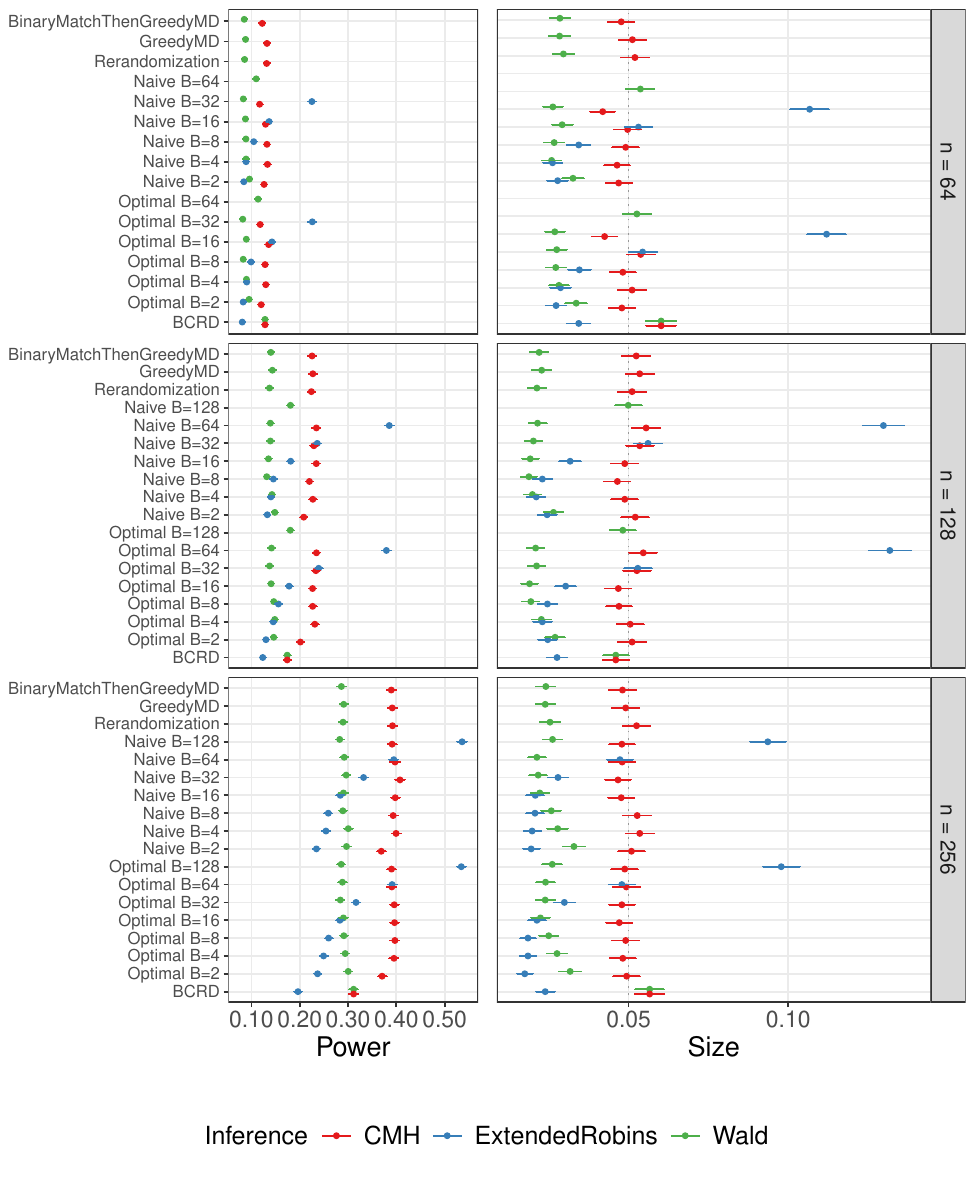}
    \caption{Power and size results for legal inference when $p=1$, all sample sizes, all designs and all three inference methods. Horizontal bars illustrate 95\% confidence bands for the simulation error. \qu{Naive} refers to the \textit{Naive Blocking} design and \qu{Optimal} refers to the \textit{Optimal Blocking} design.}
    \label{fig:power_and_size_p_1}
\end{figure}

\begin{figure}[htp]
    \centering
    \includegraphics[width=1\linewidth]{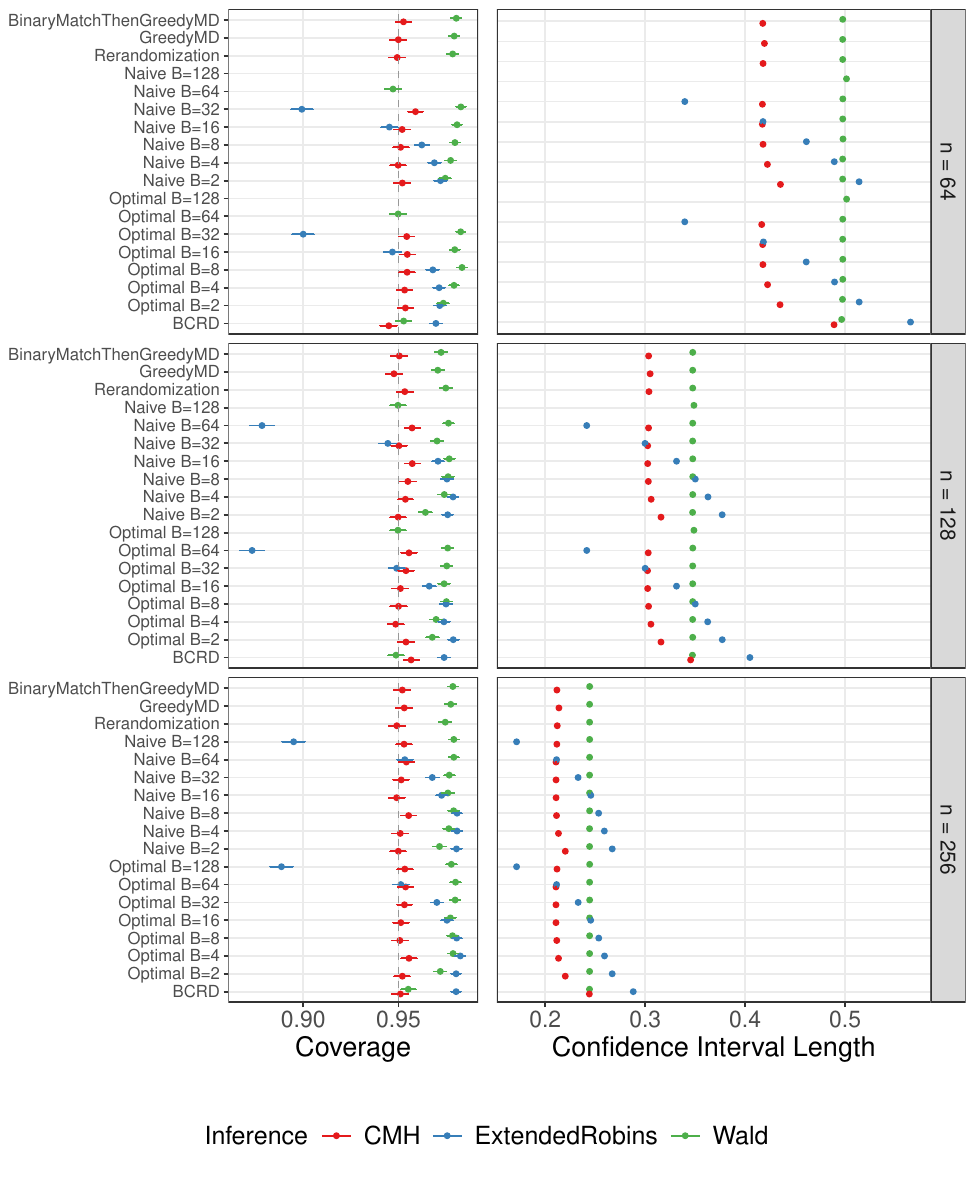}
    \caption{Coverage and confidence interval length results analogous to Figure~\ref{fig:power_and_size_p_1}.}
    \label{fig:coverage_and_length_p_1}
\end{figure}

\begin{table}[h]
\centering
\begin{tabular}{ll|rr|rr}
  \hline
Design & Inference & Power & CI & Power & CI Length \\
 & & & Length & Gain (\%) & Reduction (\%) \\
\hline
 Optimal B=2 & CMH & 11.38 & 0.42 & 53.99 & 16.28 \\ 
  Optimal B=2 & Wald & 7.39 & 0.50 &  &  \\ 
   \hline
Optimal B=4 & CMH & 12.64 & 0.40 & 82.13 & 19.35 \\ 
  Optimal B=4 & Wald & 6.94 & 0.50 &  &  \\ 
   \hline
Optimal B=8 & CMH & 12.29 & 0.40 & 98.87 & 20.58 \\ 
  Optimal B=8 & Wald & 6.18 & 0.50 &  &  \\ 
   \hline
Optimal B=16 & CMH & 12.45 & 0.40 & 99.52 & 20.83 \\ 
  Optimal B=16 & Wald & 6.24 & 0.50 &  &  \\ 
   \hline
Optimal B=32 & CMH & 11.30 & 0.40 & 70.95 & 20.79 \\ 
  Optimal B=32 & Wald & 6.61 & 0.50 &  &  \\ 
   \hline
BinaryMatch & CMH & 10.63 & 0.40 & 75.99 & 20.85 \\ 
  BinaryMatch & Wald & 6.04 & 0.50 &  &  \\ 
   \hline
Rerandomization & CMH & 12.82 & 0.40 & 84.73 & 20.17 \\ 
  Rerandomization & Wald & 6.94 & 0.50 &  &  \\ 
   \hline
GreedyMD & CMH & 12.17 & 0.40 & 84.39 & 19.67 \\ 
  GreedyMD & Wald & 6.60 & 0.50 &  &  \\ 
   \hline
\footnotesize{BinaryMatchThenGreedyMD} & CMH & 11.33 & 0.40 & 84.83 & 20.74 \\ 
  \footnotesize{BinaryMatchThenGreedyMD} & Wald & 6.13 & 0.50 &  &  \\ 
   \hline
BCRD & CMH & 11.07 & 0.49 & 0.00 & 1.64 \\ 
  BCRD & Wald & 11.07 & 0.50 &  &  \\ 
   \hline
 \hline
\end{tabular}
\caption{Results for $n=64,~p=1$ for all legal designs and inference methods. The last two columns illustrate performance gains vis-a-vis the Wald procedure. The \textit{Naive $B$} design is omitted as for $p=1$, the results are statistically equal to the \textit{Optimal $B$} design as they share the same block structure.}
\label{tab:cmh_gains}
\end{table}

\begin{figure}[htp]
    \centering
    \includegraphics[width=0.70\linewidth]{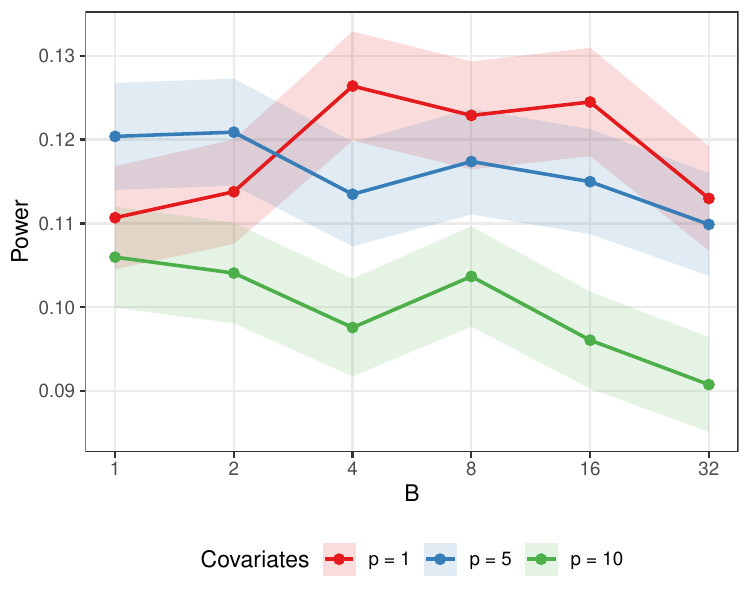}
    \caption{Power performance of the CMH estimator over different block sizes $B$ and different number of covariates in the \textit{Optimal Blocking} design for $n=64$ ($B=1$ is equivalent to the BCRD). The ribbons illustrate 95\% simulation error confidence bands.}
    \label{fig:hormesis}
\end{figure}







\section{Discussion}\label{sec:discussion}

This paper develops a unified framework for design and inference in randomized experiments with binary outcomes under Neyman’s nonparametric model where covariates are fixed and potential outcomes are random. We derive the exact variance of the difference-in-means estimator under balanced designs and show that asymptotic variance optimality is achieved by designs that satisfy a covariate-balance condition, with a broad class of blocking designs meeting this requirement under mild smoothness assumptions on the response probabilities. We also show that unbiased variance estimation is impossible in this setting, and therefore propose conservative alternatives: a CMH-based variance estimator for general balanced designs and an extension of Robins’ estimator for blocking designs. The theoretical results are complemented by simulations showing that the CMH-based procedure is generally appropriately-sized, more powerful than both the Wald-based inference and the extended Robbins one, and yields shorter confidence intervals. Blocking designs reveal a trade-off between covariate balance and design randomness that can produce a U-shaped relationship between performance and block size.

The main takeaway is that, for binary treatment effects under fixed covariates, good experimental design and valid inference cannot be separated: variance reduction depends on covariate balance, but reliable uncertainty quantification also depends on maintaining sufficient randomness in the design. The analysis clarifies that overly restrictive designs may improve balance at the cost of inflating dependence in the assignment mechanism, thereby weakening robustness and potentially harming inference. Conversely, designs that are \qu{too random} such as BCRD preserve robustness but may leave substantial precision on the table. The simulations reinforce this theoretical message by showing that intermediate blocking levels as well as other designs that achieve covariate balance and sufficient degree of randomness, perform best when they are equipped with the CMH-based inference. 

Several important directions remain for future work. First, it would be valuable to extend the results beyond binary responses and develop analogous optimal design and robust inference theory for non-binary outcomes, while carefully identifying the conditions under which the corresponding variance formulas, conservative estimators, and asymptotic tightness results continue to hold. Second, although blocking designs satisfy the boundedness conditions used in the local asymptotic analysis, it would be important to study whether similar conditions hold for broader classes of designs—specifically, to characterize when  $\lambda_{max}$   
and  $\max_{i} \braces{\sum_{j} \left| (\bSigmaw)_{i,j}\right|}$ remain bounded for designs other than blocking. More generally, a deeper understanding of how assignment covariance structure governs the trade-off between balance, randomness, and inferential accuracy could guide the construction of new experimental designs that retain the robustness of randomization while approaching optimal precision.

\bibliographystyle{apalike}
\bibliography{refs}

\clearpage
\appendix
\setcounter{section}{0} 
\addcontentsline{toc}{section}{\protect\numberline{\thesection}Appendix Title}

\begin{center}
    \Large{Supporting Information for} \\
    \large\qu{\ourtitle} \\
    \normalsize by \\
    David Azriel, Abba M. Krieger and Adam Kapelner
\end{center}

\section{Proofs}

\subsection{Computation of the Variance under Setting (c)}\label{app:variance_computation}

{\bf Note:} In this subsection{,} $\x_1,\ldots,\x_n$ are considered random variables (unlike the rest of the paper) and expectations and variances are with respect to the $\x_i$'s as well.

\begin{proposition}\label{prop:c}
	Suppose that  $\{(Y_{T,i},Y_{C,i},\x_i)\}_{i=1,\ldots,n}$ are i.i.d. from some distribution. Consider a balanced design (as defined at the beginning of Section \ref{sec:setting}){,} then
	\[
	\var{\hat{\tau}}=\frac{2}{n}\parens{p_T(1-p_T)+p_C(1-p_C)},
	\]
	where ${p}_T:=\expe{Y_{T,i}}$ and ${p}_C:=\expe{Y_{C,i}}$. 
\end{proposition}

{\bf Proof:}
First{,} notice that unconditionally,
\[
\expe{Y_{T,i}}=\expe{\cexpe{Y_{T,i}}{\x_i}}=\expe{p_{T,i}}=:p_T,
\]
and similarly, $\expe{Y_{C,i}}=:p_C$. Also,
\begin{multline*}
	\var{Y_{T,i}}
	=\expe{\cvar{Y_{T,i}}{\x_i}}+\var{\cexpe{Y_{T,i}}{\x_i}}
	=\expe{p_{T,i}(1-p_{T,i})}+\var{p_{T,i}}\\
	=\expe{p_{T,i}}-\expe{p^2_{T,i}}+\expe{p^2_{T,i}}-\parens{\expe{p_{T,i}}}^2
	=p_T(1-p_T).    
\end{multline*}
and similarly{,} $\var{Y_{C,i}}=p_C(1-p_C)$.

Now,
\begin{equation}\label{eq:var_uncond}
	\var{\hat{\tau}}
	=\expe{\cvar{\hat{\tau}}{\W}}+\var{\cexpe{\hat{\tau}}{\W}}.    
\end{equation}
We have that
\[
\cexpe{\hat{\tau}}{\W}
=\frac{1}{n}\sum_{i=1}^n W_i (p_T +p_C)+p_T -p_C.
\]
Recall that $\sum_{i=1}^n W_i=0$ and therefore $\cexpe{\hat{\tau}}{\W}$ is a constant{,} and hence 
$\var{\cexpe{\hat{\tau}}{\W}}=0$. Continuing with the first term in Equation~\ref{eq:var_uncond}, notice that conditional on $\W$, the $Y_i$'s are independent and therefore,

\beqn
	\expe{\cvar{\hat{\tau}}{\W}}
	&=&\frac{4}{n^2}\sum_{i=1}^n\expe{I(W_i=1) \var{Y_{T,i}}+I(W_i=-1) \var{Y_{C,i}} }\\
	&=&\frac{2}{n}\parens{p_T(1-p_T)+p_C(1-p_C)}.
\eeqn

\noindent Summing up,
\[
\var{\hat{\tau}}=\frac{2}{n}\parens{p_T(1-p_T)+p_C(1-p_C)}.
\]
\qed

\subsection{Proof of Theorem \ref{thm:blocking_balanced}}\label{app:proof_thm_1}

We assume with out loss of generality that the Lipschitz constant is 1 and that the support of the $\x_i$'s is contained in the p-dimensional unit cube. We have that,
\[
\frac{1}{n}(\p_T+\p_C)^\top \bSigmaw (\p_T+\p_C)=\frac{4}{B}\sum_{b=1}^B \frac{1}{n_B-1} \sum_{i \in I_b}(v_i-\bar{v}_b)^2,
\]
and by the Lipschitz assumption,
\begin{equation}\label{eq:need.to.show}
\frac{4}{B}\sum_{b=1}^B \frac{1}{n_B-1} \sum_{i \in I_b}(v_i-\bar{v}_b)^2 \le \frac{4}{B}\sum_{b=1}^B \frac{1}{n_B-1} \sum_{i \in I_b}\|\x_i-\bar{\x}_b\|^2,    
\end{equation}
where $\bar{\x}_b:=\frac{1}{n_B}\sum_{i \in I_b} \x_i$. The optimal blocking design minimizes the right-hand side of Equation~\ref{eq:need.to.show}. As in  {in} \citet[Section A.6]{Kapelner2025} we shall prove the result by considering a sub-optimal blocking algorithm, which is easy to analyze, and yields an upper bound to the right-hand side of Equation~\ref{eq:need.to.show}. To this end, we define Algorithm~\ref{alg:suboptimal_matching} below. It partitions the $p$-dimensional unit cube into $m$ equally-sized sub-cubes where $m:=m_0^p$ and $m_0:= \floor{B^{1/2p}}$. Thus, $m\approx \sqrt{B}$, and there are much more blocks than sub-cubes. Observations are assigned to the same block if they are in the same sub-cube. An \qu{overflow} group  are assigned randomly.

\begin{algorithm}
\caption{A suboptimal blocking algorithm}
\label{alg:suboptimal_matching}
~\\\textbf{Step 1}\\
Let $m_0:= \floor{B^{1/2p}}$, where $\floor{a}$ denotes the largest integer smaller than $a$, and let $m:=m_0^p$. The $p$-dimensional unit cube is partitioned into $m$ equally-sized sub-cubes, such that the edges of the sub-cubes are of length $1/m_0$. \\

\textbf{Step 2}\\
Based on the $m$ sub-cubes, $\x_i$'s are assigned into blocks such that  $i_1,i_2$ are in the same block if $x_{i_1.j}$ and $x_{i_2,j}$ are in the same sub-cube. (This implies that individuals in the same block have similar values for all the $p$ covariates.)\\


\textbf{Step 3}\\
Collect an \qu{overflow} group corresponding to all individuals who are not matched because they belong to a cell where the number of subjects after step 2 is smaller than $n_B$. 
Randomly assign all the overflow group into blocks.
\end{algorithm}

Let $\tilde{I}_b$ denote the indexes in $I_b$ that were assigned in Step 2 of Algorithm A, and $\tilde{I}_b^c:=I_b \setminus \tilde{I}_b$; thus, indexes in $\tilde{I}_b^c$ belong to the overflow group of Step 3. Let $\tilde{\x}_b$ denote the average of $\x_i$'s for $i \in \tilde{I}_b$, i.e., $\tilde{\x}_b:=\frac{1}{|\tilde{I}_b|}\sum_{i \in \tilde{I}_b} \x_i$. Because the average minimizes the ${\cal L}_2$ distance, we have that $\sum_{i \in I_b}\|\x_i-\bar{\x}_b\|^2\le \sum_{i \in I_b}\|\x_i-\tilde{\x}_b\|^2$. Therefore, going back to Equation~\ref{eq:need.to.show},
\[
\begin{aligned}
\frac{4}{B}\sum_{b=1}^B \frac{1}{n_B-1} \sum_{i \in I_b}\|\x_i-\bar{\x}_b\|^2 &\le 
\frac{4}{B}\sum_{b=1}^B \frac{1}{n_B-1} \sum_{i \in I_b}\|\x_i-\tilde{\x}_b\|^2\\
&=\frac{4}{B}\sum_{b=1}^B \frac{1}{n_B-1}\parens{\sum_{i \in \tilde{I}_b}\|\x_i-\tilde{\x}_b\|^2+\sum_{i \in \tilde{I}_b^c}\|\x_i-\tilde{\x}_b\|^2}.
\end{aligned}
\]
For indexes $i$ in $\tilde{I}_b$, $|x_{i,j}-x_{b,j}|\le \frac{1}{m_0}$ for all $j=1,\ldots,p$, by construction. Also, $|\tilde{I}_b|\le n_B$ because $n_B$ is the block size. Thus, $\sum_{i \in I_b}\|\x_i-\tilde{\x}_b\|^2 \le \frac{p n_B}{m_0^2}$. 
For indexes $i$ in $\tilde{I}_b^c$, $\|\x_i-\tilde{\x}_b\|^2 \le p$ since the $\x_i$'s are in the p-dimensional unit cube. Further, the size of the overflow group is bounded by $m (n_B-1)$, because for each sub-cube there are at most $n_B-1$ observations that are not matched at Step 2 of Algorithm A. Hence, 
$\sum_{b=1}^B \sum_{i \in \tilde{I}_c} \|\x_i-\tilde{\x}_b\|^2  \le p m(n_B-1)$. It follows that,
\[
\frac{4}{B}\sum_{b=1}^B \frac{1}{n_B-1}\parens{\sum_{i \in \tilde{I}_b}\|\x_i-\tilde{\x}_b\|^2+\sum_{i \in \tilde{I}_b^c}\|\x_i-\tilde{\x}_b\|^2}
\le \frac{4}{B}\sum_{b=1}^B \frac{pn_B}{m_0^2(n_B-1)} +4p \frac{m}{B}\le 8p\parens{\frac{1}{m_0^2}+\frac{m}{B}}.
\]
We have that $B^{1/2p} -1 \le m_0 \le B^{1/2p}$ and recall that $m=m_0^p$. Hence,
\[
\frac{1}{m_0^2}+\frac{m}{B} \le \frac{1}{(B^{1/2p} -1)^2}+\frac{1}{\sqrt{B}},
\]
which converges to zero as $B\to \infty$. \qed

\subsection{Proof of Theorem \ref{thm:exp_var_hat}}\label{app:thm_2_proof}

We have that
\[
\expe{ \Y^\top \bSigmaw \Y}
= \left( 
\sum_{i_1 \ne i_2} \expe{Y_{i_1} Y_{i_2}} (\bSigmaw)_{i_1,i_2}
+\sum_i \expe{Y_{i}^2}  
\right)
\]
because the diagonal elements of $\bSigmaw$ are equal to 1.
Now,
\[
\expe{Y_{i_1} Y_{i_2}}
= \expe{\cexpe{Y_{i_1} Y_{i_2}}{\W}}.
\]
Given $\W$, $Y_{i_1}$ and $Y_{i_2}$ are independent{,} and Equation~\ref{eq:Y_obs} implies that 
$\cexpe{Y_i}{W_{i}}={v_i + W_{i} \eta_i}$,  
where $v_i:=\frac{p_{T,i}+p_{C,i}}{2}$ and $\eta_i:=\frac{p_{T,i}-p_{C,i}}{2}$ and{,} recall{,} $\expe{W_i}=0$. Therefore,
\[
\expe{\cexpe{Y_{i_1} Y_{i_2}}{\W}}
= v_{i_1} v_{i_2}
+ \expe{W_{i_1} W_{i_2}}\eta_{i_1} \eta_{i_2}
= v_{i_1} v_{i_2} + (\bSigmaw)_{i_1,i_2} {\eta_{i_1} \eta_{i_2}}.
\]
Also, since $Y_i \in \{0,1\}$,
\[
\expe{Y_i^2}= \expe{Y_i}= \expe{\cexpe{Y_i}{W_{B,i}}}=v_i.
\]

Hence,
\begin{multline}\label{eq:termmm}
\expe{ \Y^\top \bSigmaw \Y}
=  \sum_{i_1 \ne i_2} v_{i_1} v_{i_2} (\bSigmaw)_{i_1,i_2}
+\sum_i v_i  
+ \sum_{i_1 \ne i_2} \eta_{i_1} \eta_{i_2} (\bSigmaw)_{i_1,i_2}^2 \\
=  \v^\top \bSigmaw \v 
	+ \v^\top(\onevec-\v) 
	+ {\boldsymbol \eta}^\top \bSigmaw^{\odot 2} {\boldsymbol \eta} 
	-  \normsq{\bseta},
\end{multline}
where $\bSigmaw^{\odot 2}:=\bSigmaw \odot \bSigmaw$ and $\odot$ denotes element-wise product.

Now, by straightforward algebra we have
\begin{equation*}\label{eq:v(1-v)}
v_i(1-v_i) 
= \frac{1}{2} [p_{T,i}(1-p_{T,i})+p_{C,i}(1-p_{C,i})]
	+\frac{(p_{T,i}-p_{C,i})^2 }{4}
=\frac{1}{2} [p_{T,i}(1-p_{T,i})+p_{C,i}(1-p_{C,i})]+\eta_i^2.
\end{equation*}
Therefore,
\[
\v^\top(\onevec-\v)  -  \normsq{\bseta}
=\frac{1}{2}[\p_T^\top(\onevec -\p_T) + \p_C^\top(\onevec -\p_C)].
\]
Also,
\[
\v^\top \bSigmaw \v
=\frac{1}{4} (\p_T+\p_C)^\top \bSigmaw (\p_T+\p_C).
\]
Therefore, Equation~\ref{eq:termmm} implies that
\begin{equation*}
\begin{aligned}
\expe{V_{CMH}}
&=\frac{4}{n^2}\expe{ \Y^\top \bSigmaw \Y}\\
&=
\frac{1}{n^2}\parens{
(\p_T+\p_C)^\top \bSigmaw (\p_T+\p_C)
+2 [\p_T^\top(\onevec -\p_T) + \p_C^\top(\onevec -\p_C)]
+4  {\boldsymbol \eta}^\top \bSigmaw^{\odot 2} {\boldsymbol \eta}
}\\
&=\var{\hat{\tau}}
+ \frac{4}{n^2}  {\boldsymbol \eta}^\top \bSigmaw^{\odot 2} {\boldsymbol \eta}, 
\end{aligned}
\end{equation*}
where the last equality is due to Equation~\ref{eq:var}.
Schur product theorem implies that $\bSigmaw^{\odot 2}$ is positive semi-definite because $\bSigmaw$ is positive semi-definite. Hence, $\expe{V_{CMH}}\ge \var{\hat{\tau}}$.
\qed

\subsection{Proof of Proposition \ref{prop:BCRD}}\label{app:proof_prop_1}

Under BCRD{,} the extended Robins estimator for the variance is
\begin{multline}\label{eq:V_Robbins_adj}
V_{Robbins-ext}
=V_{Robbins} + \frac{1}{n} \left( \hat{p}_T(1-\hat{p}_T)+ \hat{p}_C(1-\hat{p}_C)\right)\\
=\frac{m_1(1-m_1)}{n/2}
+\frac{m_0(1-m_0)}{n/2}
+\frac{(2m_0-m_1)(1-m_1)-m_0(1-m_0)}{n}
+\frac{m_1(1-m_1)}{n}
+\frac{m_0(1-m_0)}{n}\\
= \frac{2}{n}\left[ {m_1(1-m_1)}+{m_0(1-m_0)} + m_0(1-m_1)\right]\\
= \frac{2}{n}\left[  \hat{p}_T(1-\hat{p}_T)+ \hat{p}_C(1-\hat{p}_C) + m_0(1-m_1)\right].
\end{multline}

The variance of the CMH is
\[
V_{CMH}
=\frac{4}{n^2}\Y^\top \bSigmaw \Y 
= \frac{4}{n(n-1)}\sum_{i=1}^n(Y_i - \bar{Y})^2
=\frac{4}{n-1} \bar{Y}(1-\bar{Y}),
\]
where the last equality holds because $Y_i^2=Y_i$ and hence 
$\frac{1}{n}\sum_{i=1}^n(Y_i - \bar{Y})^2=\bar{Y}(1-\bar{Y})$.
In terms of $\hat{p}_T$ and $\hat{p}_C$, we have
\[
V_{CMH}
=\frac{4}{n}\frac{1}{n-1}\sum_{i=1}^n(Y_i - \bar{Y})^2
=\frac{4}{n-1}\,\frac{\hat{p}_T+\hat{p}_C}{2}
\left(1- \frac{\hat{p}_T+\hat{p}_C}{2}\right).
\]
Let $\rho:=\frac{\hat{p}_T+\hat{p}_C}{2}$.  
As in Equation~\ref{eq:v(1-v)}, we have
\[
\rho(1-\rho)
=\frac{1}{2}\!\parens{\hat{p}_T(1-\hat{p}_T)+\hat{p}_C(1-\hat{p}_C)}
+\frac{(\hat{p}_T-\hat{p}_C)^2}{4}.
\]
Therefore,
\begin{equation}\label{eq:V_CMH}
V_{CMH}
=\frac{2}{n-1}\!\left[\hat{p}_T(1-\hat{p}_T)+ \hat{p}_C(1-\hat{p}_C)\right]
+ \frac{ (\hat{p}_T - \hat{p}_C)^2}{n-1}.
\end{equation}
\qed

\subsection{Proof of Theorem \ref{thm:local}}\label{app:proof_thm3}

By Theorem \ref{thm:exp_var_hat},
\[
n \parens{\expe{V_{CMH}}- \var{\hat{\tau}}}=\frac{4}{n}{\boldsymbol \eta}^\top \bSigmaw^{\odot 2} {\boldsymbol \eta}
\]
and by the Fan-Horn inequality \citep[][page 312]{Horn_Johnson_1991},
\[
{\boldsymbol \eta}^\top \bSigmaw^{\odot 2} {\boldsymbol \eta} \le \lambda_{max}^2(\bSigmaw) \| {\boldsymbol \eta}\|^2. 
\]
It follows that
\[
n \parens{\expe{V_{CMH}}- \var{\hat{\tau}}} \le 4 \lambda_{max}^2(\bSigmaw)\frac{1}{n} \| {\boldsymbol \eta}\|^2.
\]
We assumed that $\frac{1}{n} \| {\boldsymbol \eta}\|^2 \to 0$ (Equation~\ref{eq:moment_condition}) and that $\lambda_{max}(\bSigmaw)$ is bounded, hence, $n \parens{\expe{V_{CMH}}- \var{\hat{\tau}}}$ converges to zero.

To prove the moreover part, we bound $\var{n V_{CMH}}$.
We have that 
\[
\var{2 n V_{CMH}}=\var{ \frac{1}{n} \Y^\top \bSigmaw \Y} = \expe{\cvar{ \frac{1}{n} \Y^\top \bSigmaw \Y}{\W}} + \var{\cexpe{ \frac{1}{n} \Y^\top \bSigmaw \Y}{\W}}.
\]
Given $\W$, the $Y_i$'s are independent. In a previous work \citep{Azriel2024} we have calculated $\var{\Z^\top \bSigmaw \Z}$ where $\Z$ has independent entries but not identically distributed (with mean zero). Based on these computations (see Eq. (32) in \citet{Azriel2024}) it can be shown that $\var{ \frac{1}{n} \Y^\top \bSigmaw | \W} \le C/n$ for a constant $C$ when the moments of the $Y_i$'s and  $\frac{1}{n} \tr{\bSigmaw^2}$ are bounded, which is the case in our setting since $Y_i \in \{0,1\}$ and $\frac{1}{n} \tr{\bSigmaw^2} \le \lambda_{max}^2(\bSigmaw)$, which we assumed is bounded.

Consider now $ \var{\cexpe{ \frac{1}{n} \Y^\top \bSigmaw \Y}{\W}}$. Because $\cexpe{Y_i}{W_{i}}={v_i + W_{i} \eta_i}$ (recall that $v_i:=\frac{p_{T,i}+p_{C,i}}{2}$ and $\eta_i:=\frac{p_{T,i}-p_{C,i}}{2}$) and $Y_i^2=Y_i$, then
\begin{equation*}
\begin{aligned}
\cexpe{ \frac{1}{n} \Y^\top \bSigmaw \Y}{\W}&= \frac{1}{n} \sum_{i_1 \ne i_2} (v_{i_1}+W_{i_1} \eta_{i_1}) (v_{i_2}+W_{i_2} \eta_{i_2})+  \frac{1}{n}\sum_{i}(v_i + W_{i} \eta_i)\\
&= \frac{1}{n}\sum_{i_1 \ne i_2} (v_{i_1}+W_{i_1} \eta_{i_1}) (v_{i_2}+W_{i_2} \eta_{i_2})+  \frac{1}{n}\sum_{i} v_i,    
\end{aligned}
\end{equation*}
where the last equation is true because $\sum_{i=1}^n W_i=0$.
It follows that
\begin{multline*}
 \var{\cexpe{ \frac{1}{n} \Y^\top \bSigmaw \Y}{\W}}\\
 =\frac{1}{n^2} \sum_{i_1 \ne i_2,i_3 \ne i_4}
 \Big\{ \cov{ (v_{i_1}+W_{i_1} \eta_{i_1})
(v_{i_2}+W_{i_2} \eta_{i_2})}{(v_{i_3}+W_{i_3} \eta_{i_3}) (v_{i_4}+W_{i_4} \eta_{i_4})}\\  (\bSigmaw)_{i_1,i_2} (\bSigmaw)_{i_3,i_4}\Big\}.
\end{multline*}
Notice that the case we excluded $i_1=i_2,i_3=i_4$ is non-negative because
\[
\begin{aligned}
&\sum_{i_1,i_3}   \cov{ (v_{i_1}+W_{i_1} \eta_{i_1})^2
}{(v_{i_3}+W_{i_3} \eta_{i_3})^2} (\bSigmaw)_{i_1,i_1} (\bSigmaw)_{i_3,i_3}\\
&= \sum_{i_1,i_3}   \cov{ v_{i_1}^2 +2 v_{i_1} \eta_{i_1} W_{i_1}+ \eta_{i_1}^2
}{ v_{i_3}^2 +2 v_{i_3} \eta_{i_3} W_{i_3}+ \eta_{i_3}^2}\\
&= 4 \sum_{i_1,i_3} v_{i_1} \eta_{i_1}  v_{i_3} \eta_{i_3} \cov{  W_{i_1}
}{ W_{i_3}}=4 {\boldsymbol{\xi}}^\top \bSigmaw {\boldsymbol{\xi}} \ge 0,
\end{aligned}
\]
where $\xi_i=v_i \eta_i$; the term is non-negative because $\bSigmaw$ is positive semi-definite. It follows that
\begin{multline}\label{eq:Var_E}
 \var{\cexpe{ \frac{1}{n} \Y^\top \bSigmaw \Y}{\W}}\\
 \le \frac{1}{n^2} \sum_{i_1,i_2,i_3,i_4}
 \Big\{ \cov{ (v_{i_1}+W_{i_1} \eta_{i_1})
(v_{i_2}+W_{i_2} \eta_{i_2})}{(v_{i_3}+W_{i_3} \eta_{i_3}) (v_{i_4}+W_{i_4} \eta_{i_4})}\\  (\bSigmaw)_{i_1,i_2} (\bSigmaw)_{i_3,i_4}\Big\}.
\end{multline}
Let $\bbSigmaw$ and $\bar{\bseta}$ denote the entry-wise absolute value of $\bSigmaw$ and ${\bseta}$, respectively, and let $M:=\max_{i} \sum_{j} (\bbSigmaw)_{i,j}$. By the Gershgorin circle theorem \citep[Theorem 1.1]{varga2004gershgorin}, $\lambda_{max}(\bbSigmaw)\le M$. Notice also that $0 \le v_i \le 1$, $|\eta_i| \le 1$ for all $i$, and that $\left|(\bbSigmaw)_{i,j}\right|\le 1$ for all $i,j$, and that $|\cov{W_{i_1}^{j_1} W_{i_2}^{j_2}}{W_{i_3}^{j_3} W_{i_4}^{j_4}}| \le 2$ for all $i_1,i_2,i_3,i_4$ and for all non-negative integers $j_1,j_2,j_3,j_4$.
We now show that the nine terms of the right-hand side of Equation~\ref{eq:Var_E} converge to zero:
\begin{itemize}
\item I:
\[
\begin{aligned}
& \frac{1}{n^2} \sum_{i_1 ,i_2,i_3 , i_4} \cov{ v_{i_1}
W_{i_2} \eta_{i_2}}{v_{i_3} W_{i_4} \eta_{i_4}}(\bSigmaw)_{i_1,i_2} (\bSigmaw)_{i_3,i_4}\\
&= \frac{1}{n^2} \sum_{i_1 ,i_2,i_3 , i_4} v_{i_1} \eta_{i_2} v_{i_3} \eta_{i_4}(\bSigmaw)_{i_2,i_4} (\bSigmaw)_{i_1,i_2} (\bSigmaw)_{i_3,i_4}\\
&\le \frac{1}{n^2} \sum_{i_1 ,i_2,i_3 , i_4}  \bar{\eta}_{i_2} \bar{\eta}_{i_4} (\bbSigmaw)_{i_2,i_4}  (\bbSigmaw)_{i_3,i_4}\\
&\le \frac{M}{n} \sum_{i_2, i_4}  \bar{\eta}_{i_2} \bar{\eta}_{i_4} (\bbSigmaw)_{i_2,i_4} \\
&=\frac{M}{n} \bar{\bseta}^\top \bbSigmaw \bar{\bseta}\\
&\le M \lambda_{max}(\bbSigmaw) \frac{1}{n}\| \bseta\|^2\\
&\le M^2 \frac{1}{n}\| \bseta\|^2,
\end{aligned}
\]
which converges to zero due to Equation~\ref{eq:moment_condition} and boundness of $M$.
\item II:
\[
\frac{1}{n^2} \sum_{i_1 ,i_2,i_3 , i_4} \cov{ v_{i_1}
W_{i_2} \eta_{i_2}}{ W_{i_3} \eta_{i_3}v_{i_4}}(\bSigmaw)_{i_1,i_2} (\bSigmaw)_{i_3,i_4}.
\]
Converges to zero by the same argument as I.
\item III:
\[
\begin{aligned}
& \frac{1}{n^2} \sum_{i_1 ,i_2,i_3 , i_4} \cov{ v_{i_1}
W_{i_2} \eta_{i_2}}{W_{i_3} \eta_{i_3} W_{i_4} \eta_{i_4}}(\bSigmaw)_{i_1,i_2} (\bSigmaw)_{i_3,i_4}\\
&= \frac{1}{n^2} \sum_{i_1 ,i_2,i_3 , i_4} v_{i_1} \eta_{i_2} \eta_{i_3} \eta_{i_4} \cov{ 
W_{i_2}}{W_{i_3} W_{i_4} } (\bSigmaw)_{i_1,i_2} (\bSigmaw)_{i_3,i_4}\\
&\le \frac{1}{n^2} \sum_{i_1 ,i_2,i_3 , i_4}  \bar{\eta}_{i_3} \bar{\eta}_{i_4} (\bbSigmaw)_{i_1,i_2}  (\bbSigmaw)_{i_3,i_4}\\
&\le \frac{M}{n} \sum_{i_3, i_4}  \bar{\eta}_{i_3} \bar{\eta}_{i_4} (\bbSigmaw)_{i_3,i_4} \\
&=\frac{M}{n} \bar{\bseta}^\top \bbSigmaw \bar{\bseta}\\
&\le M \lambda_{max}(\bbSigmaw) \frac{1}{n}\| \bseta\|^2\\
&\le M^2 \frac{1}{n}\| \bseta\|^2,
\end{aligned}
\]
which converges to zero due to Equation~\ref{eq:moment_condition} and boundness of $M$.
\item IV:
\[
 \frac{1}{n^2} \sum_{i_1 ,i_2,i_3 , i_4} \cov{ W_{i_1} \eta_{i_1} v_{i_2}
}{v_{i_3} W_{i_4} \eta_{i_4}}(\bSigmaw)_{i_1,i_2} (\bSigmaw)_{i_3,i_4}.
\]
Converges to zero by the same argument as I.
\item V:
\[
 \frac{1}{n^2} \sum_{i_1 ,i_2,i_3 , i_4} \cov{ W_{i_1} \eta_{i_1} v_{i_2}
}{W_{i_3} \eta_{i_3} v_{i_4}}(\bSigmaw)_{i_1,i_2} (\bSigmaw)_{i_3,i_4}.
\]
Converges to zero by the same argument as I.
\item VI:
\[
 \frac{1}{n^2} \sum_{i_1 ,i_2,i_3 , i_4} \cov{ W_{i_1} \eta_{i_1} v_{i_2}
}{W_{i_3} \eta_{i_3} W_{i_4} \eta_{i_4}}(\bSigmaw)_{i_1,i_2} (\bSigmaw)_{i_3,i_4}.
\]
Converges to zero by the same argument as III.
\item VII:
\[
 \frac{1}{n^2} \sum_{i_1 ,i_2,i_3 , i_4} \cov{ W_{i_1} \eta_{i_1}  W_{i_2} \eta_{i_2}
}{v_{i_3} W_{i_4} \eta_{i_4}}(\bSigmaw)_{i_1,i_2} (\bSigmaw)_{i_3,i_4}.
\]
Converges to zero by the same argument as III.
\item VIII:
\[
 \frac{1}{n^2} \sum_{i_1 ,i_2,i_3 , i_4} \cov{ W_{i_1} \eta_{i_1}  W_{i_2} \eta_{i_2}
}{W_{i_3} \eta_{i_3} v_{i_4} }(\bSigmaw)_{i_1,i_2} (\bSigmaw)_{i_3,i_4}.
\]
Converges to zero by the same argument as III.
\item IX:
\[
\begin{aligned}
& \frac{1}{n^2} \sum_{i_1 ,i_2,i_3 , i_4} \cov{ W_{i_1} \eta_{i_1}
W_{i_2} \eta_{i_2}}{W_{i_3} \eta_{i_3} W_{i_4} \eta_{i_4}}(\bSigmaw)_{i_1,i_2} (\bSigmaw)_{i_3,i_4}\\
&= \frac{1}{n^2} \sum_{i_1 ,i_2,i_3 , i_4} \eta_{i_1} \eta_{i_2} \eta_{i_3} \eta_{i_4} \cov{W_{i_1} 
W_{i_2}}{W_{i_3} W_{i_4} } (\bSigmaw)_{i_1,i_2} (\bSigmaw)_{i_3,i_4}
\\
&\le \frac{2}{n^2} \sum_{i_1 ,i_2,i_3 , i_4} \bar{\eta}_{i_1} \bar{\eta}_{i_2} \bar{\eta}_{i_3} \bar{\eta}_{i_4}(\bbSigmaw)_{i_1,i_2} (\bbSigmaw)_{i_3,i_4}\\
&=\frac{2}{n^2} \parens{\bar{\bseta}^\top \bbSigmaw \bar{\bseta}}^2\\
&\le 2\lambda_{max}(\bbSigmaw)^2 \parens{\frac{1}{n}\| \bseta\|^2}^2\\
&\le 2 M^2 \parens{\frac{1}{n}\| \bseta\|^2}^2,
\end{aligned}
\]
which converges to zero due to Equation~\ref{eq:moment_condition} and boundness of $M$.
\end{itemize}
The proof of the theorem is thus concluded.

\qed

\section{Additional Simulation Results}\label{app:additional_results}

\begin{figure}[htp]
    \centering
    \includegraphics[width=1\linewidth]{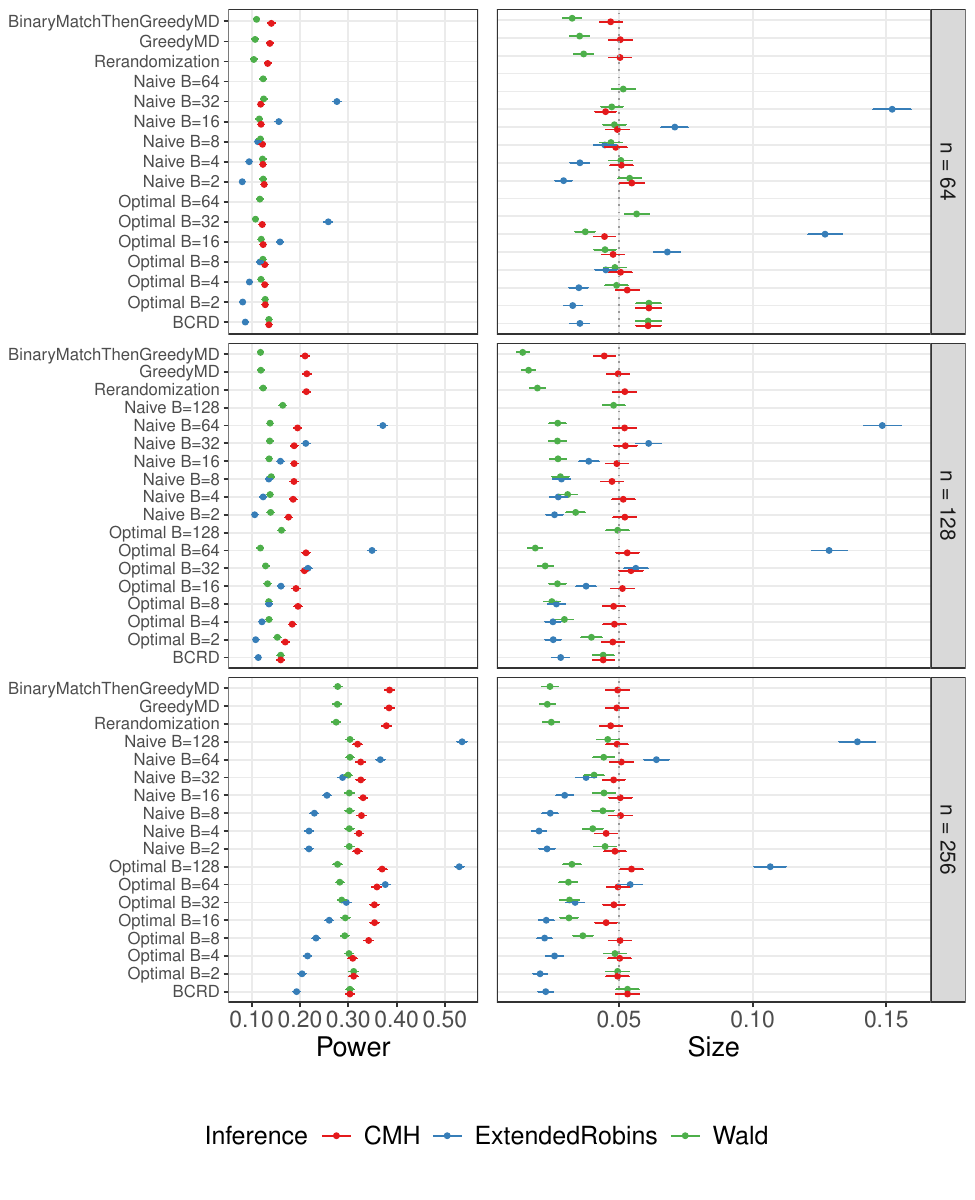}
    \caption{Power and size results for legal inference when $p=5$.}
    \label{fig:power_and_size_p_5}
\end{figure}

\begin{figure}[htp]
    \centering
    \includegraphics[width=1\linewidth]{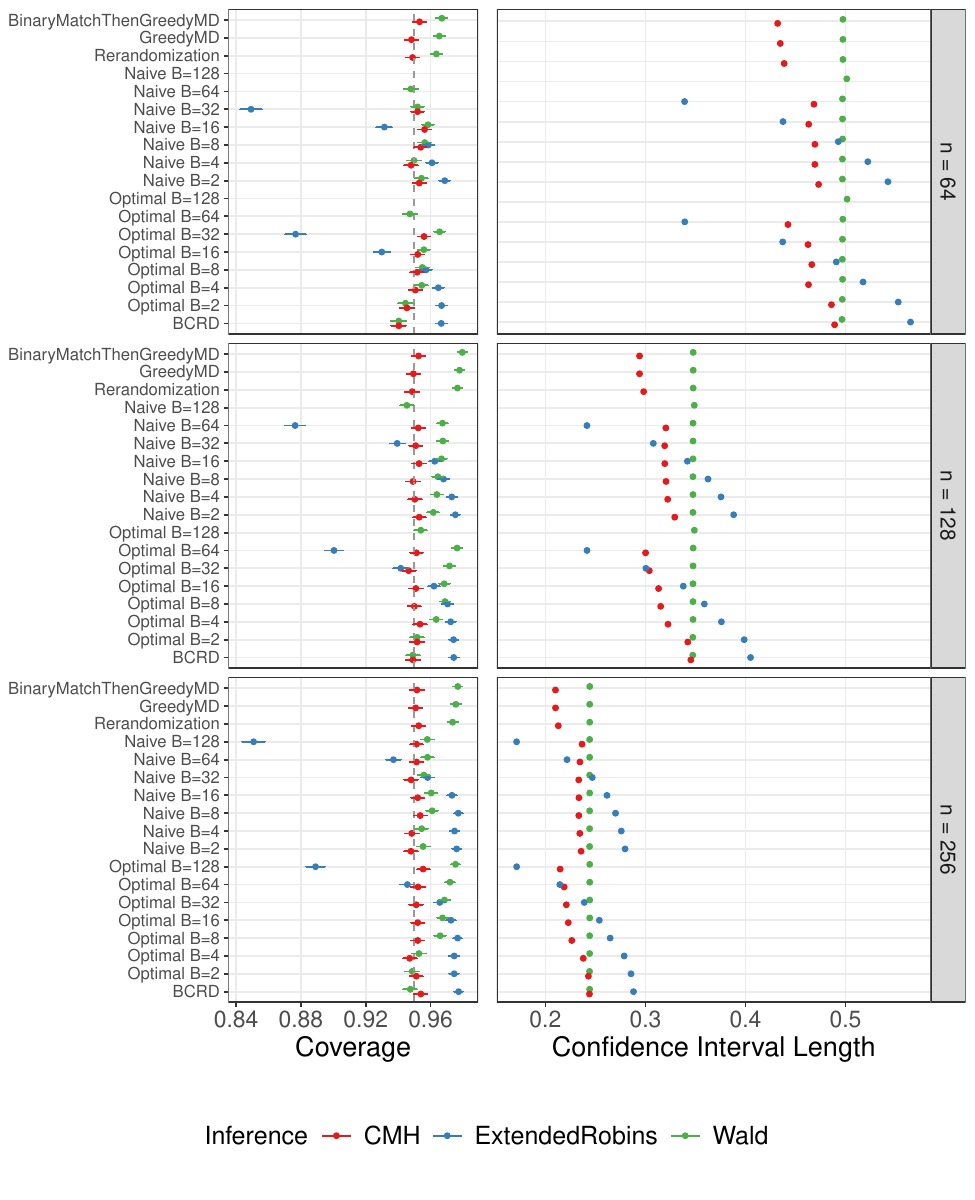}
    \caption{Coverage and confidence interval length results for legal inference when $p=5$.}
    \label{fig:coverage_and_length_p_5}
\end{figure}

\begin{figure}[htp]
    \centering
    \includegraphics[width=1\linewidth]{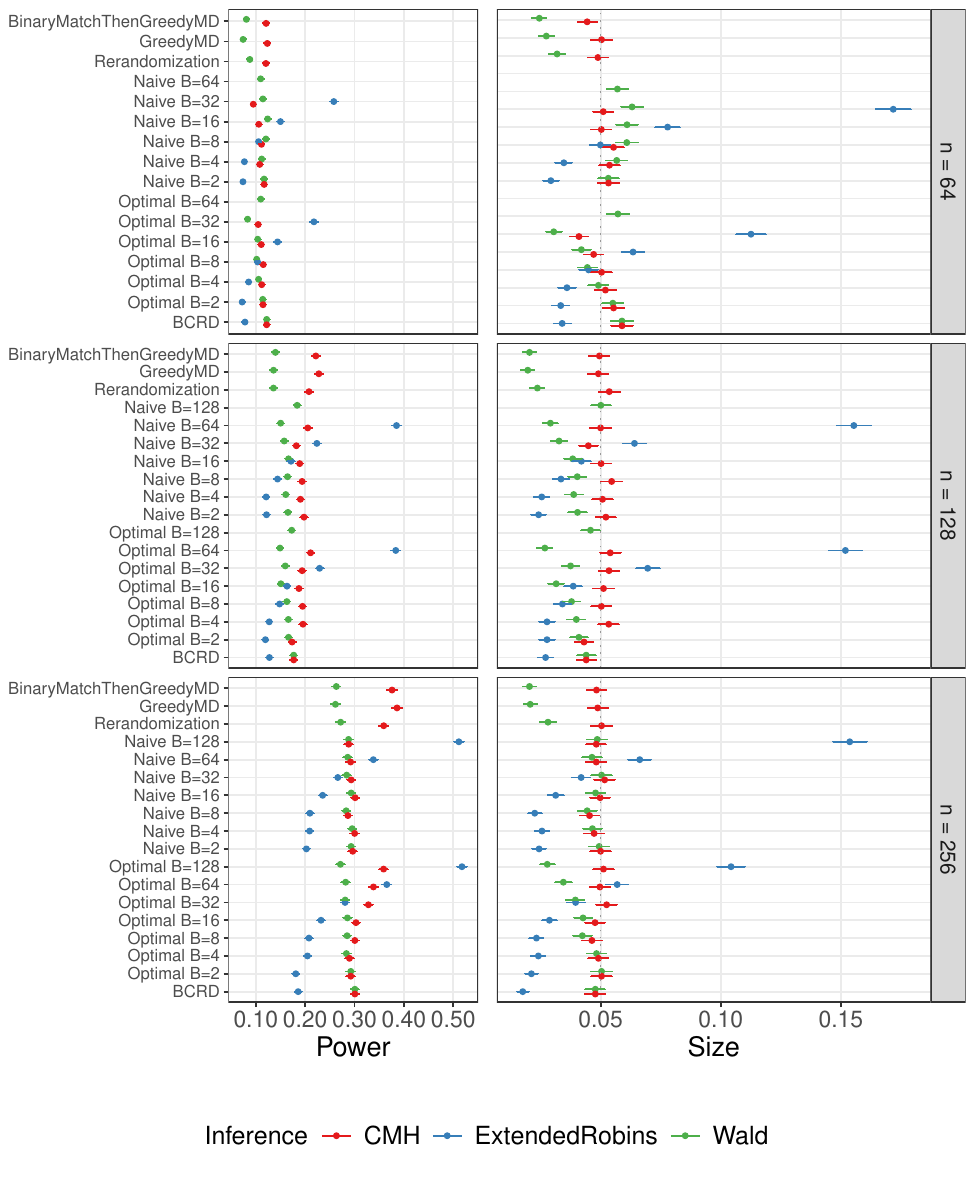}
    \caption{Power and size results for legal inference when $p=10$.}
    \label{fig:power_and_size_p_10}
\end{figure}

\begin{figure}[htp]
    \centering
    \includegraphics[width=1\linewidth]{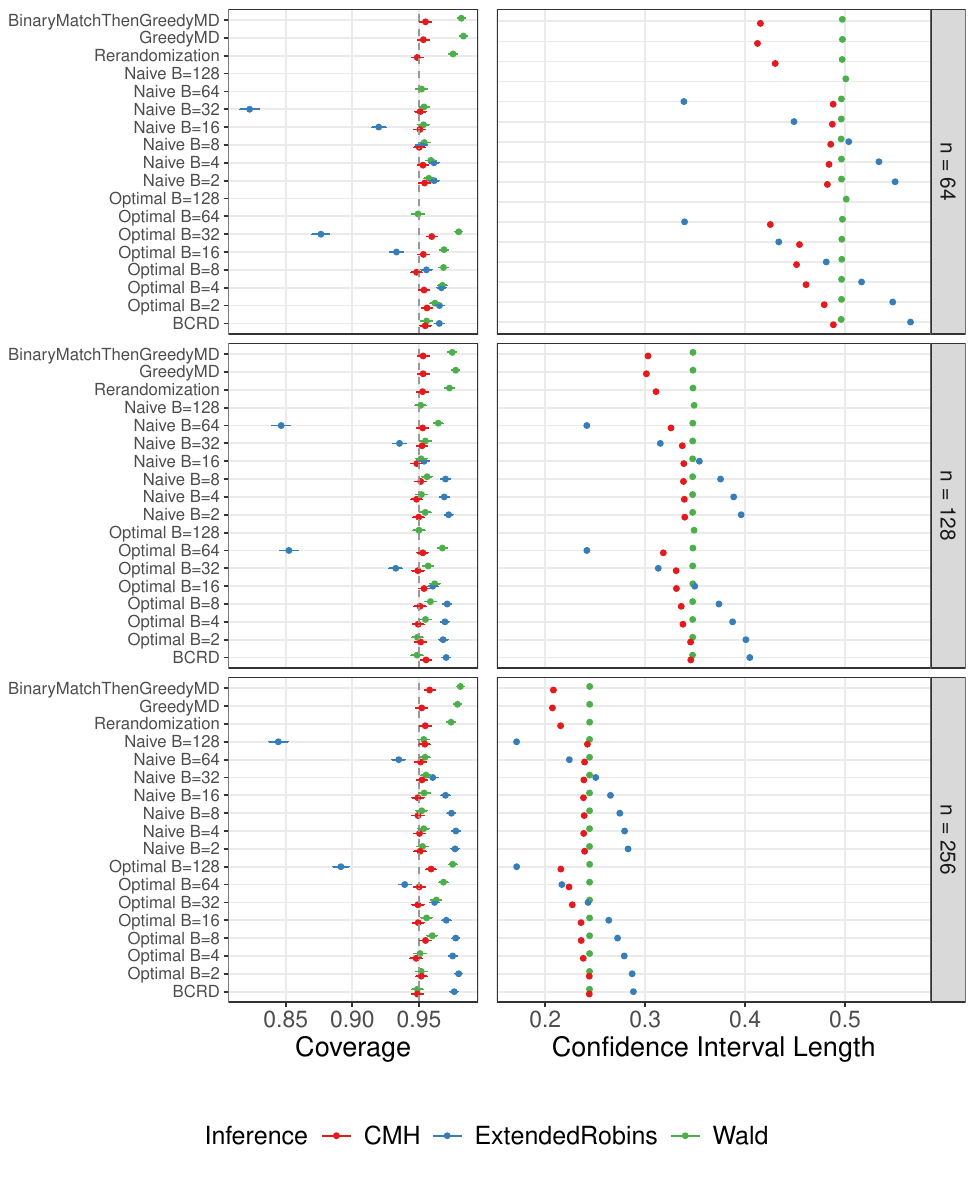}
    \caption{Coverage and confidence interval length results for legal inference when $p=10$.}
    \label{fig:coverage_and_length_p_10}
\end{figure}

\end{document}

%% file: preamble.tex
\usepackage{lmodern}
\usepackage{graphicx}
\usepackage{subcaption}
\usepackage{amsmath}
\usepackage{dsfont}
\usepackage{placeins}
\usepackage{multirow}
\usepackage{tabularray}
\usepackage{hyperref}
\usepackage{enumerate}

\usepackage{xcolor}
\usepackage{algpseudocode}
\usepackage{amsthm}
\usepackage{hyperref}
\usepackage{mathrsfs}
\usepackage{amsfonts}
\usepackage{bbding}
\usepackage{listings}
\usepackage{appendix}
\usepackage{fancyvrb}
\usepackage{cancel}
\usepackage{algorithm}

\newcounter{probnum}
\allowdisplaybreaks

\definecolor{tabblue}{rgb}{.870588,.905882,.94902}
\definecolor{gray}{rgb}{0.7,0.7,0.7}
\definecolor{black}{rgb}{0,0,0}
\definecolor{white}{rgb}{1,1,1}
\definecolor{blue}{rgb}{0.0,0.0,1}

\definecolor{green}{rgb}{0,0.5,0}

\definecolor{yellow}{rgb}{1,0.549,0}

\definecolor{red}{rgb}{0.6,0.0,0.0}

\definecolor{darkred}{rgb}{0.9,0.4,0}

\definecolor{purple}{rgb}{0.58,0,0.827}

\definecolor{backgcode}{rgb}{0.97,0.97,0.8}
\definecolor{Brown}{cmyk}{0,0.81,1,0.60}
\definecolor{OliveGreen}{cmyk}{0.64,0,0.95,0.40}
\definecolor{CadetBlue}{cmyk}{0.62,0.57,0.23,0}

\newcommand{\qu}[1]{``{#1}''}
\newcommand{\bv}[1]{\boldsymbol{#1}}

\newcommand{\bSigma}{\bv{\Sigma}}

\newcommand{\bSigmaw}{\bv{\Sigma}_{\W}}
\newcommand{\bbSigmaw}{\bar{\bv{\Sigma}}_{\W}}

\newcommand{\bseta}{\boldsymbol{\eta}}

\newcommand{\ybar}{\bar{y}}
\newcommand{\Ybar}{\bar{Y}}

\newcommand{\Z}{\bv{Z}}
\newcommand{\X}{\bv{X}}

\newcommand{\p}{\bv{p}}

\newcommand{\Y}{\bv{Y}}
\newcommand{\W}{\bv{W}}

\newcommand{\x}{\bv{x}}

\newcommand{\w}{\bv{w}}

\newcommand{\tauhat}{\hat{\tau}}

\newcommand{\onevec}{\bv{1}}
\newcommand{\zerovec}{\bv{0}}

\newcommand{\y}{\bv{y}}

\renewcommand{\v}{\bv{v}}

\newcommand{\bbeta}{\bv{\beta}}

\newcommand{\reals}{\mathbb{R}}

\newcommand{\floor}[1]{\left\lfloor #1 \right\rfloor}

\newcommand{\beqn}{\vspace{-0.25cm}\begin{eqnarray*}}
\newcommand{\eeqn}{\end{eqnarray*}}
\newcommand{\bneqn}{\vspace{-0.25cm}\begin{eqnarray}}
\newcommand{\eneqn}{\end{eqnarray}}
\newcommand{\benum}{\begin{enumerate}}
\newcommand{\eenum}{\end{enumerate}}

\newcommand{\parens}[1]{\left(#1\right)}

\newcommand{\prob}[1]{\mathbb{P}\parens{#1}}

\newcommand{\bracks}[1]{\left[#1\right]}
\newcommand{\braces}[1]{\left\{#1\right\}}

\newcommand{\norm}[1]{\left|\left|#1\right|\right|}
\newcommand{\normsq}[1]{\norm{#1}^2}

\newcommand{\expe}[1]{\mathbb{E}\bracks{#1}}

\newcommand{\cexpe}[2]{\expe{#1\,|\,#2}}

\newcommand{\cvar}[2]{\var{#1\,|\,#2}}

\newcommand{\var}[1]{\mathbb{V}\text{ar}\bracks{#1}}

\newcommand{\expenostr}[1]{\mathbb{E}[#1]}

\newcommand{\cov}[2]{\mathbb{C}\text{ov}\bracks{#1,\,#2}}

\newcommand{\oneover}[1]{\frac{1}{#1}}

\newcommand{\convd}{~{\buildrel d \over \longrightarrow}~}



%% file: refs.bib
@article{Krieger2022,
  title={Better experimental design by hybridizing binary matching with imbalance optimization},
  author={Krieger, Abba M and Azriel, David A and Kapelner, Adam},
  journal={Canadian Journal of Statistics},
  year={2022},
  doi={10.1002/cjs.11685}
}

@article{Kapelner2023,
  title={The role of pairwise matching in experimental design for an incidence outcome},
  author={Kapelner, Adam and Krieger, Abba M and Azriel, David},
  journal={Australian \& New Zealand Journal of Statistics},
  volume={65},
  number={4},
  pages={379--393},
  year={2023}
}

@article{Morgan2012,
  title={Rerandomization to improve covariate balance in experiments},
  author={Morgan, K L and Rubin, D B},
  journal={The Annals of Statistics},
  pages={1263--1282},
  year={2012}
}

@article{Kallus2018,
  title={Optimal a priori balance in the design of controlled experiments},
  author={Kallus, Nathan},
  journal={Journal of the Royal Statistical Society: Series B (Statistical Methodology)},
  volume={80},
  number={1},
  pages={85--112},
  year={2018},
  publisher={Wiley Online Library}
}

@article{Student1938,
  title={Comparison between balanced and random arrangements of field plots},
  author={Student},
  journal={Biometrika},
  pages={363--378},
  year={1938},
  publisher={JSTOR}
}

@article{Freedman2008,
author = {Freedman, D A},
journal = {Advances in Applied Mathematics},
month = {feb},
number = {2},
pages = {180--193},
title = {{On regression adjustments to experimental data}},
volume = {40},
year = {2008}
}

@article{Li2016,
  title={Asymptotic Theory of Rerandomization in Treatment-Control Experiments},
  author={Li, X and Ding, P and Rubin, D B},
  journal={arXiv preprint arXiv:1604.00698},
  year={2016}
}

@article{Krieger2019,
  title={Nearly random designs with greatly improved balance},
  author={Krieger, Abba M and Azriel, David and Kapelner, Adam},
  journal={Biometrika},
  volume={106},
  number={3},
  pages={695--701},
  year={2019}
}

@article{Stuart2010,
  title={Matching methods for causal inference: A review and a look forward},
  author={Stuart, Elizabeth A},
  journal={Statistical science},
  volume={25},
  number={1},
  pages={1},
  year={2010}
}

@book{Imbens2015,
  title={Causal inference in statistics, social, and biomedical sciences},
  author={Imbens, Guido W and Rubin, Donald B},
  year={2015},
  publisher={Cambridge University Press}
}

@article{Kapelner2021,
  title={Harmonizing Optimized Designs with Classic Randomization in Experiments},
  author={Kapelner, Adam and Krieger, Abba M and Sklar, Michael and Shalit, Uri and Azriel, David},
  journal={The American Statistician},
  volume={75},
  number={2},
  pages={195--206},
  year={2021},
  publisher={Taylor \& Francis}
}

@article{Cochran1954,
  title={Some methods for strengthening the common $\chi^2$ tests},
  author={Cochran, William G},
  journal={Biometrics},
  volume={10},
  number={4},
  pages={417--451},
  year={1954},
  publisher={JSTOR}
}

@article{Mantel1959,
  title={Statistical aspects of the analysis of data from retrospective studies of disease},
  author={Mantel, Nathan and Haenszel, William},
  journal={Journal of the national cancer institute},
  volume={22},
  number={4},
  pages={719--748},
  year={1959},
  publisher={Oxford University Press}
}

@misc{Azriel2024,
      title={The Optimality of Blocking Designs in Equally and Unequally Allocated Randomized Experiments with General Response}, 
      author={David Azriel and Abba M. Krieger and Adam Kapelner},
      year={2024},
      journal={arXiv preprint arXiv:2212.01887},
      url={https://arxiv.org/abs/2212.01887}
}

@article{Kapelner2025,
  title={The Pairwise Matching Design Is Optimal Under Extreme Noise and Extreme Assignments},
  author={Kapelner, Adam and Krieger, Abba M and Azriel, David},
  journal={Stat},
  volume={14},
  number={2},
  pages={e70058},
  year={2025},
  publisher={Wiley Online Library}
}

@article{Robins1988,
  title={Confidence intervals for causal parameters},
  author={Robins, James M},
  journal={Statistics in medicine},
  volume={7},
  number={7},
  pages={773--785},
  year={1988},
  publisher={Wiley Online Library}
}

@incollection{Sekhon2008,
  author    = {Sekhon, Jasjeet S.},
  title     = {The {Neyman--Rubin} Model of Causal Inference and Estimation via Matching Methods},
  booktitle = {The Oxford Handbook of Political Methodology},
  editor    = {Box-Steffensmeier, Janet M. and Brady, Henry E. and Collier, David},
  publisher = {Oxford University Press},
  address   = {New York},
  year      = {2008},
  pages     = {271--299},
  doi       = {10.1093/oxfordhb/9780199286546.003.0011},
  isbn      = {9780199286546}
}

@article{Rigdon2015,
  title={Randomization inference for treatment effects on a binary outcome},
  author={Rigdon, Joseph and Hudgens, Michael G},
  journal={Statistics in medicine},
  volume={34},
  number={6},
  pages={924--935},
  year={2015},
  publisher={Wiley Online Library}
}

@inproceedings{Aronow2023,
  title={Fast computation of exact confidence intervals for randomized experiments with binary outcomes},
  author={Aronow, PM and Chang, Haoge and Lopatto, Patrick},
  booktitle={Proceedings of the 24th ACM Conference on Economics and Computation},
  pages={120--120},
  year={2023}
}

@article{Li2025exact,
  title={Exact and Conservative Inference for the Average Treatment Effect in Stratified Experiments with Binary Outcomes},
  author={Li, Jiaxun and Spertus, Jacob and Stark, Philip B},
  journal={arXiv preprint arXiv:2508.03834},
  year={2025}
}

@article{Bai2024,
  title={A primer on the analysis of randomized experiments and a survey of some recent advances},
  author={Bai, Yuehao and Shaikh, Azeem M and Tabord-Meehan, Max},
  journal={arXiv preprint arXiv:2405.03910},
  year={2024}
}

@InProceedings{Malinen2014,
author="Malinen, Mikko I.
and Fr{\"a}nti, Pasi",
editor="Fr{\"a}nti, Pasi
and Brown, Gavin
and Loog, Marco
and Escolano, Francisco
and Pelillo, Marcello",
title="Balanced K-Means for Clustering",
booktitle="Structural, Syntactic, and Statistical Pattern Recognition",
year="2014",
publisher="Springer Berlin Heidelberg",
address="Berlin, Heidelberg",
pages="32--41",
isbn="978-3-662-44415-3"
}

@article{Edmonds1965,
  author  = {Jack Edmonds},
  title   = {Paths, Trees, and Flowers},
  journal = {Canadian Journal of Mathematics},
  year    = {1965},
  volume  = {17},
  pages   = {449--467},
  doi     = {10.4153/CJM-1965-045-4},
}

@article{Neyman1990,
  title={On the application of probability theory to agricultural experiments. Essay on principles. Section 9},
  author={Splawa-Neyman, Jerzy and Dabrowska, Dorota M and Speed, Terrence P},
  journal={Statistical Science},
  pages={465--472},
  year={1990},
  publisher={JSTOR}
}

@article{azriel2026block,
  title={Block designs that provide optimal power in the Cochran--Mantel--Haenszel test: D. Azriel et al.},
  author={Azriel, David and Kapelner, Adam and Krieger, Abba M},
  journal={Statistical Papers},
  volume={67},
  number={2},
  pages={28},
  year={2026},
  publisher={Springer}
}

@book{Horn_Johnson_1991, place={Cambridge}, title={Topics in Matrix Analysis}, publisher={Cambridge University Press}, author={Horn, Roger A. and Johnson, Charles R.}, year={1991}}

@article{Harshaw2024,
author = {Christopher Harshaw and Fredrik Sävje and Daniel A. Spielman and Peng Zhang},
title = {Balancing Covariates in Randomized Experiments with the Gram–Schmidt Walk Design},
journal = {Journal of the American Statistical Association},
volume = {119},
number = {548},
pages = {2934--2946},
year = {2024},
publisher = {Taylor \& Francis},
doi = {10.1080/01621459.2023.2285474},
}

@book{Van2000,
  title={Asymptotic statistics},
  author={Van der Vaart, Aad W},
  volume={3},
  year={2000},
  publisher={Cambridge university press}
}

@book{varga2004gershgorin,
  title={Gersgorin and His Circles},
  author={Varga, Richard S.},
  series={Springer Series in Computational Mathematics},
  volume={36},
  year={2004},
  publisher={Springer},
  address={Berlin, Heidelberg},
  isbn={978-3-540-21100-6},
  doi={10.1007/978-3-642-17798-9}
}

@book{horn2013matrix,
  title     = {Matrix Analysis},
  author    = {Horn, Roger A. and Johnson, Charles R.},
  edition   = {2},
  year      = {2013},
  publisher = {Cambridge University Press},
  address   = {Cambridge},
  isbn      = {9780521548236}
}

@article{Derigs1988,
    author  = {Ulrich Derigs},
    title   = {Solving Non-Bipartite Matching Problems via Shortest Path Techniques},
    journal = {Annals of Operations Research},
    year    = {1988},
    volume  = {13},
    number  = {1},
    pages   = {225--261},
    doi     = {10.1007/BF02288324}
  }

@Manual{Beck2024nbpMatching,
    title  = {{nbpMatching}: Functions for Optimal Non-Bipartite Matching},
    author = {Cole Beck and Bo Lu and Robert Greevy},
    year   = {2024},
    note   = {R package version 1.5.6},
    url    = {https://CRAN.R-project.org/package=nbpMatching},
    doi    = {10.32614/CRAN.package.nbpMatching}
  }

@Manual{Papenberg2026,
    title  = {{anticlust}: Subset Partitioning via Anticlustering},
    author = {Martin Papenberg},
    year   = {2026},
    note   = {R package version 0.8.14},
    url    = {https://CRAN.R-project.org/package=anticlust},
    doi    = {10.32614/CRAN.package.anticlust}
  }
